\documentclass[a4paper,11pt]{article}
\usepackage{jcapmod}
\usepackage[a4paper, top=2.5cm, bottom=2.5cm, left=2cm, right=2cm]{geometry}

\usepackage[utf8]{inputenc}
\usepackage[english]{babel}
\usepackage{graphicx}
\usepackage{bm}
\usepackage{mathtools}
\usepackage{lipsum}
\usepackage{amssymb}
\usepackage{dsfont}
\usepackage{float}
\usepackage{booktabs}
\usepackage{array}
\usepackage[table]{xcolor}
\usepackage{tabularx}
\usepackage{caption}

\usepackage{xcolor}
\usepackage{xspace}
\usepackage{ifthen}
\usepackage[font=scriptsize]{caption}
\graphicspath{{graphics}}

\usepackage[unicode=true,pdfusetitle,
 bookmarks=true,bookmarksnumbered=false,bookmarksopen=false,
 breaklinks=false,pdfborder={0 0 1},backref=false,colorlinks=true]
 {hyperref}
\hypersetup{
 linkcolor=blue, citecolor=cyan, urlcolor=cyan, filecolor=blue}

\usepackage{titlesec}
\titleformat{\subsubsection}
  {\normalfont\normalsize\itshape}{\thesubsubsection}{1em}{}

\usepackage{titling}
\usepackage{amsthm}
\newtheorem{theorem}{Theorem}

\newtheorem{definition}{Definition}

\usepackage{comment}
\usepackage{tikz}
\usetikzlibrary{shapes,arrows.meta,positioning,calc,shadows}
\usepackage{float}

\newcommand{\gui}[1]%
{\ifthenelse{\equal{\showcomments}{true}}%
{{\color{blue}{\small [\textbf{G:} #1}]}}{\xspace}}%

\newcommand{\nicolas}[1]%
{\ifthenelse{\equal{\showcomments}{true}}%
{{\color{orange}{\small [\textbf{N:} #1}]}}{\xspace}}%

\newcommand{\hilbert}{\mathcal{H}}

\newcommand{\alga}{\mathcal{A}}
\newcommand{\algb}{\mathcal{B}}
\newcommand{\algm}{\mathcal{M}}

\newcommand{\idm}{\mathbb{I}}

\newtheorem{prop}[theorem]{Proposition}

\newtheorem{ex}[theorem]{Example}
\newtheorem{aspt}[theorem]{Assumption}

\newcommand{\Hcal}{\mathcal{H}\xspace}

\DeclareMathOperator{\tr}{tr}

\newcommand{\ket}[1]{| #1 \rangle}
\newcommand{\bra}[1]{\langle #1 |}

\newcommand{\braket}[2]{\langle #1 | #2 \rangle}

\newcommand{\showcomments}{true}

\theoremstyle{nonitalic}

\begin{document} 

\pagenumbering{roman}

\begin{titlepage}

\baselineskip=15.5pt \thispagestyle{empty}


\vspace{-2cm}

\begin{center}
    {\fontsize{19}{24}\selectfont \bfseries Subsystems (in)dependence in GIE proposals}
\end{center}


\vspace{0.1cm}

\renewcommand*{\thefootnote}{\fnsymbol{footnote}}

\begin{center}
     Nicolas Boulle$^{1}$ \footnote{nicolas.boulle@miun.se},   Guilherme Franzmann$^{2,3,4}$ \footnote{guilherme.franzmann@su.se}
\end{center}

\renewcommand*{\thefootnote}{\arabic{footnote}}
\setcounter{footnote}{1}

\vspace{-0.3cm}
\begin{center}
\rule{14cm}{0.08cm}
\end{center}

\footnotesize
\vskip4pt
   \noindent
    \vskip4pt
   \noindent 
    \textsl{$^1$Department of Engineering, Mathematics, and Science Education, Mid Sweden University, Holmgatan 10, 851 70 Sundsvall, Sweden}

    \vskip4pt
   \noindent 
    \textsl{$^2$Basic Research Community for Physics e.V., Mariannenstraße 89, Leipzig, Germany}

    \vskip4pt
   \noindent 
    \textsl{$^3$Nordita, KTH Royal Institute of Technology and Stockholm University,
Hannes Alfvéns v\"ag 12, SE-106 91 Stockholm, Sweden}

    \vskip4pt
   \noindent 
    \textsl{$^4$Department of Philosophy, Stockholm University, Stockholm, Sweden}

\normalsize


\vspace{0.5cm}

\hrule
\vspace{0.3cm}

\noindent {\bf Abstract} 

Recent proposals suggest that detecting entanglement between two spatially superposed masses would establish the quantum nature of gravity. However, these gravitationally induced entanglement (GIE) experiments rely on assumptions about subsystem independence. We sharpen the theoretical underpinnings of such proposals by examining them through the lens of algebraic quantum field theory (AQFT), distinguishing operational and algebraic notions of independence. We argue that state and measurement independence of subsystems, essential to the experimental logic, is nontrivial in the presence of gauge constraints and gravitational dressing. Using gravitationally dressed fields, we recall that commutation relations between spacelike separated observables are nontrivial, undermining strict Hilbert space factorization. We further explore the implications for entanglement witnesses, investigating the Tsirelson bound when subsystem algebras fail to commute, and showing that the bound persists for a suitably symmetrized CHSH observable even though the operational status of such ``joint'' observables becomes delicate when commutativity is lost. Our analysis highlights how even within linearized covariant quantum gravity, violations of microcausality may affect both the interpretation, modelling, and design of proposed laboratory tests of quantum gravity, despite remaining negligible for current experimental regimes. Although we consider GIE-style protocols as a concrete case study, the subsystem-independence issues we highlight are generic to low-energy (perturbative) quantum gravity. Finally, we derive estimates for dressing-induced microcausality violations, which suggest a complementary avenue to current proposals as a probe of the quantum nature of gravity (likely far beyond current experimental sensitivity, though).

\vskip10pt
\hrule
\vskip10pt

\newpage
\setlength{\topmargin}{0in}

\end{titlepage}

\clearpage



\tableofcontents

\noindent\makebox[\linewidth]{\rule{\textwidth}{0.4pt}}

\pagenumbering{arabic}
\setcounter{page}{1}

\section{The quantum nature of gravity}

Quantum mechanics is widely regarded as the fundamental paradigm describing nature, yet gravity has stubbornly resisted it. While other interactions are successfully described by quantum field theories, gravity can be consistently treated within this framework only as an effective field theory at low energies due to its perturbative non-renormalizability \cite{Wallace:2021qyh}. For this reason,  whether gravity is truly a quantum field or whether an alternative ontology is more appropriate remains an open question. Arguably \cite{Fragkos:2022tbm}, empirical evidence has so far probed gravity directly only in the classical regime (acting on quantum matter), leaving the quantum nature of spacetime itself untested. Indeed, attempts to retain a fundamentally classical gravity consistently coupled to quantum matter have been put forward, most prominently in the ``post-quantum'' framework of \cite{Oppenheim:2018igd}.

In recent years, new proposals have emerged to test the quantum character of gravity using tabletop experiments. A major motivation is that gravitationally induced entanglement (GIE) would constitute a genuinely quantum-gravitational signature in a regime accessible to near-term experiments \cite{Bose,DelicEtAl2020LevitatedCooling,WestphalEtAl2021MillimeterGravity,MargalitEtAl2021SternGerlachInterferometer,Aspelmeyer2022AvoidClassicalWorld,PandaEtAl2024LatticeAtomInterferometerGravity,bose2025spinbasedpathwaytestingquantum}. These ideas are rooted in quantum information–theoretic results: specifically, that local operations and classical communication (LOCC) (see \cite{Chitambar:2014svb} for a review) cannot generate entanglement, while a quantum mediator can. The seminal proposals by Bose et al. \cite{Bose} and Marletto–Vedral \cite{Marletto} (BMV) suggest that detecting entanglement between two spatially superposed masses, interacting solely via gravity, would imply that the mediator---the gravitational interaction---possesses quantum degrees of freedom. A widely advertised strengthening of this claim is that the inference can be made largely model-independent, by appealing to entanglement monotonicity under LOCC and related ``no-go'' results for entanglement generation via purely classical mediators \cite{Horodecki:2009zz,Galley:2020qsf,ludescher2025gravitymediatedentanglementinfinitedimensionalsystems}. This has sparked extensive analysis and debate (see also \cite{Christodoulou:2022mkf,Huggett:2022uui,Fragkos:2022tbm}) about what such experiments truly test and under which assumptions. In particular, the scope of the LOCC-style inference (e.g. what exactly counts as ``classical'' mediation, and which locality assumptions are doing the work) has been repeatedly contested in the recent literature \cite{Huggett:2022uui,Fragkos:2022tbm,DiBiagio:2025twt}. 

However, these schemes are not truly model-independent. Even when expressed in a Newtonian nonlocal Hamiltonian form, the dynamics tacitly assume linearized covariant quantum gravity, i.e., a spin-2 massless field mediating the interaction \cite{Bose_Supplementary_2017}. Therefore, we believe it is within the quantum-field theoretic framework that the generation of entanglement and the structure of entanglement witnesses must be evaluated\footnote{Recently it has been suggested that classical gravity alone as a field theory coupled with quantum-field-theory matter (emphasis on their field nature) can evade the LOCC restriction and generate entanglement \cite{Aziz:2025ypo}, while \cite{ludescher2025gravitymediatedentanglementinfinitedimensionalsystems} claims exactly the opposite by modeling both systems as unital C$^*$-algebras; gravity as a commutative one. One possible reconciliation might be due to the different notions of locality (spatiotemporal or in Hilbert space) that are often conflated; see \cite{DiBiagio:2025twt} for a concise discussion.}. Thus, our work aims to sharpen this discussion by analyzing subsystem independence in this context. We recall that the BMV-type protocols assume initially uncorrelated systems, represented in a factorized Hilbert space. Yet, in quantum field theory with long-range gauge fields, and especially in gravity, algebraic notions of independence that underlie Hilbert-space factorizability are nontrivial, and both state and measurement independence might be compromised. In fact, dressed field operators, required to make observables gauge-invariant, generally do not commute at spacelike separation in linearized quantum gravity \cite{Donnelly:2015hta,Donnelly:2016rvo}. As we will see, that alone is enough to undermine the exact Hilbert-space factorization presupposed in the quantum-information description and the existence of uncorrelated states, and thus challenge a clear notion of quantum subsystems \cite{franzmann2024be,EmerGe_proj_2}.   

To scrutinize the modelling of the BMV-type proposals, we will resort to several theorems proved in the context of algebraic quantum field theory. In particular, we will rely on the thorough work by Summers \cite{summers1990,summers2009} where notions of independence were deeply investigated and clarified. To make the discussion more digestible, we will focus only on the most relevant results that have direct bearings on the assumptions of the experiment and its possible interpretations.  

Having a clear understanding of the experimental protocol’s underpinnings extends beyond the direct interpretation of positive or negative detection. For instance, some authors argue that a positive BMV-type signal would constitute evidence for quantum superpositions of spacetime geometries \cite{Christodoulou_2019}. This interpretive step, however, is conditional: it presupposes that the weak-field limit of GR remains valid for micron-scale configurations of nanometre-sized test masses, and it relies on a particular quantum-field-theoretic modelling of the gravitational mediator. As our analysis shows, questions of subsystem independence, gauge-invariant dressing, and potential microcausality issues all bear directly on how entanglement witnesses are implemented and how their outcomes should be read. Until these modelling choices and their implications are established, such ontological claims remain suggestive interpretations of gravity as a quantum mediator, rather than unambiguous evidence for ``spacetime superposition''.

Finally, within linearized covariant quantum gravity, signatures of quantumness may arise beyond entanglement. In particular, we estimate the size of microcausality violations induced by gravitational dressing, and we translate these effects into order-of-magnitude, dimensionless figures of merit for parameter regimes relevant to BMV-type setups. While these effects are expected to be extremely small, they provide a concrete target for assessing how
``locality'' and subsystem structure are modified once one insists on gauge-invariant observables.

The paper is divided as follows. We start by reviewing the different accounts of the BMV experiments in Sec. \ref{sec:BMV}, making their underlying assumptions explicit.  In Sec.~\ref{section:notions_indep}, we introduce and compare different notions of subsystem independence in quantum theory, emphasizing the distinction between commensurability, statistical independence, and the split property, and highlighting their equivalence in finite-dimensional quantum mechanics and their inequivalence in quantum field theory. In Sec.~\ref{section:Dressed states and Commutativity}, we analyze gauge-invariant dressed observables in QED and linearized quantum gravity, showing that while commutativity at spacelike separation is preserved in the former, it is generically violated in the latter. In Sec.~\ref{sec:bell}, we examine Tsirelson’s bound in the presence of noncommuting subalgebras, clarifying why standard entanglement witnesses remain formally valid. Sec.~\ref{sec:microcausality_violation} then identifies dressing-induced violations of microcausality as a potential signature of linearized quantum gravity and provides order-of-magnitude estimates relevant for current experimental parameters. We conclude by relating our results to other, conceptually similar obstructions to subsystem factorization in gravitational settings, and by briefly commenting on implications for inflationary cosmology.

\paragraph{Operational vs.\ algebraic notions.}
A standing theme in what follows is a distinction between \emph{operational} and \emph{algebraic} claims about subsystems. Operational claims concern what can be established from laboratory statistics, and are formulated in terms of admissible preparation and measurement procedures. No-signalling and local preparability are of this kind. In the same spirit, we take, for example, \emph{statistical independence} to primarily mean an \emph{operational} notion: the ability of different parties to prepare states independently by local procedures, without the choice of preparation by any party constraining what can be prepared by the others. Algebraic claims, by contrast, are purely structural properties of the underlying observable algebras (e.g. commutativity, commutants, inclusions, existence of an intermediate type~I factor as in the funnel property, or subalgebra type-ness), and do not by themselves encode which operations are physically available. We will therefore be explicit about which steps in our argument are algebraic and which require operational input. Finally, when we invoke the \emph{existence of uncorrelated states}, we will treat it as a formal/state-space condition whose validity may follow from algebraic structure (notably split inclusions) together with the operational assumptions being discussed\footnote{GF thanks Jan Glowacki for discussions on these matters.}.

\section{Algebraic setup and assumptions for BMV/GIE}\label{sec:BMV}

We begin by making explicit---and sharply separating---the \emph{operational} assumptions about preparation and measurement independence that underwrite BMV/GIE-style arguments, from the \emph{algebraic} conditions typically used to represent (or ``implement'') those assumptions within quantum theory. Our aim is not merely to review the proposals of ~\cite{Bose, Marletto}, but to lay out a clean logical scaffold for the subsystem structure they presuppose. We do so in a quantum-informational, algebraic language. We start with the bipartite description, then turn to the tripartite description (where gravity is modeled as an additional mediating subsystem), and finally move to the field-theoretic setting that is meant to underwrite, and potentially obstruct, both idealized subsystem pictures.

\subsection{Bipartite Description}

Consider two masses, \(m_1\) and \(m_2\), forming a joint quantum system described by a Hilbert space \(\hilbert\) with algebra of bounded observables \(\mathcal{A}=\mathcal{B}(\hilbert)\). The system is in a (normal) state \(\omega:\mathcal{A}\to\mathbb{C}\), represented by a density operator \(\rho\) such that
\begin{equation}
    \omega(O)= \tr (\rho\,O), \qquad O\in\mathcal{A}\,.
\end{equation}
Within \(\mathcal{A}\), we model the two masses by identifying two subalgebras \(\mathcal{A}_1,\mathcal{A}_2\subset\mathcal{A}\), each containing the observables associated with one of the masses. Given a global state \(\omega\), its \emph{restriction} to the algebra \(\mathcal{A}_i\) is defined as the linear functional
\begin{equation}
\omega_i := \omega|_{\mathcal{A}_i}, \qquad \text{defined by } \omega_i(A_i)\equiv \omega(A_i),\ \ \forall\,A_i\in\mathcal{A}_i\,.
\end{equation}
Operationally, this means that expectation values of observables belonging to the subalgebra \(\mathcal{A}_i\) are computed directly from the same global density operator~\(\rho\). Later, once a tensor product structure is introduced, this restriction will coincide with the usual notion of taking the partial trace.

\medskip

\noindent\textbf{Operational assumptions and their algebraic implementation.}
The first three assumptions we use are \emph{operational}: they are intended to capture, in plain terms, how the two masses can be \emph{prepared} and \emph{measured} (as evidenced by laboratory statistics). We then \emph{implement} these operational assumptions within the algebraic framework by imposing corresponding conditions on the chosen subalgebras \(\mathcal{A}_1,\mathcal{A}_2\) and on the available global state space.

\begin{aspt}[\textbf{Operational:} Independent preparation]
\label{aspt:prep-ind}
By acting locally on each mass, one can prepare them independently: any pair of local states can be realized simultaneously.
\end{aspt}

\noindent
We encode this algebraically by requiring that, for any pair of normal states \(\omega_1\) on \(\mathcal{A}_1\) and \(\omega_2\) on \(\mathcal{A}_2\), there exists at least one global normal state \(\omega\) on \(\mathcal{A}\) whose restrictions coincide with them:
\begin{equation}
\omega|_{\mathcal{A}_1}=\omega_1, \qquad \omega|_{\mathcal{A}_2}=\omega_2\,.
\end{equation}
This captures preparation independence at the level of state extensions, without yet imposing any condition on correlations or on joint measurability.

\begin{aspt}[\textbf{Operational:} Uncorrelated initial preparation]
\label{aspt:uncorr}
There exists an initial preparation in which the joint statistics contain no correlations between the two masses: outcomes on one side carry no information about the state prepared on the other.
\end{aspt}

\noindent
In the algebraic setting, the absence of correlations is expressed by the existence of (normal) \emph{uncorrelated} states: for any pair of local preparations \((\omega_1,\omega_2)\) there exists a global state \(\omega\) such that
\begin{equation}
\omega(A_1A_2)=\omega_1(A_1)\,\omega_2(A_2),
\qquad
\forall\,A_i\in\mathcal{A}_i\,.
\end{equation}
Saying that the two systems are uncorrelated thus means that the global state assigns expectation values to joint observables that decompose into a \textit{product} of their marginal expectations defined in terms of the local states. In ordinary quantum mechanics, this corresponds to a product\footnote{This factorization condition is strictly stronger than separability. In finite-dimensional quantum mechanics, a separable state has the form \(\rho=\sum_k p_k\,\rho^{(k)}_1\otimes\rho^{(k)}_2\), so that \(\omega(A_1A_2)=\sum_k p_k\,\omega^{(k)}_1(A_1)\omega^{(k)}_2(A_2)\), which in general does \emph{not} equal \(\omega_1(A_1)\omega_2(A_2)\). Separable states can therefore still carry \emph{classical} correlations, whereas the product condition above excludes such correlations.} density matrix \(\rho=\rho_1\otimes\rho_2\). Operationally, however, the intended content is prior to any sort of product structure in Hilbert space, and the product condition above is treated as a state-space condition that may (in other contexts) follow from algebraic structure (e.g.\ split inclusions) together with operational input.

\begin{aspt}[\textbf{Operational:} Independent measurability]
\label{aspt:comm-op}
The observables associated with the two masses can be measured jointly (without mutual disturbance): the measurement choices on one side do not obstruct, nor alter the statistics of, measurements performed on the other.
\end{aspt}

\noindent
A standard sufficient way to represent this operational commensurability within the algebraic framework is to require that the corresponding subalgebras commute\footnote{%
Note that the factorization condition in Assumption~\ref{aspt:uncorr}, written as $\omega(A_1A_2)=\omega_1(A_1)\omega_2(A_2)$, has a direct operational interpretation as ``absence of correlations in joint statistics'' only once one assumes commensurability (e.g.\ via $[\mathcal A_1,\mathcal A_2]=0$; see Sec. \ref{sec:bell}). Without commensurability, we treat it as a formal/state-space condition on a global extension $\omega$ rather than as an experimentally accessible joint expectation value.},
\begin{equation}
[\mathcal{A}_1,\mathcal{A}_2]=0\,.
\end{equation}
In what follows, we will be explicit when we use \emph{operational} commensurability and when we invoke its \emph{algebraic} implementation via commutativity.

\medskip

The remaining assumption is \emph{purely algebraic}: it is not an operational statement about preparations or measurements, but rather a structural condition that lets one recover the standard subsystem picture in terms of a Hilbert-space tensor product structure.

\begin{aspt}[\textbf{Algebraic:} Finite-dimensionality and completeness]
\label{aspt:finite}
The relevant subsystem algebras are finite dimensional and complete.\footnote{Completeness: the subalgebras \(\{\mathcal{A}_\alpha\}_\alpha\) jointly generate all of \(\mathcal{B}(\hilbert)\),
\(
\bigvee_\alpha \mathcal{A}_\alpha \cong \mathcal{B}(\hilbert),
\)
where \(\bigvee_\alpha \mathcal{A}_\alpha\) denotes the von Neumann algebra generated by the union of the \(\mathcal{A}_\alpha\), defined as the double commutant of the \(*\)-algebra generated by \(\cup_\alpha \mathcal{A}_\alpha\).}
\end{aspt}

\noindent
Under the \emph{algebraic} assumptions of commuting subalgebras (as an implementation of independent measurability) together with them representing local-accessible observables,  finite-dimensionality, and completeness, it follows from the result of Zanardi \emph{et al.}~\cite{Zanardi:2004zz} that the global Hilbert space admits a tensor product structure (TPS),
\begin{equation}
\hilbert \simeq \hilbert_1 \otimes \hilbert_2\,,
\end{equation}
such that the subsystem algebras can be represented as
\begin{equation}
\mathcal{A}_1 \simeq \mathcal{B}(\hilbert_1)\otimes \idm,\qquad
\mathcal{A}_2 \simeq \idm \otimes \mathcal{B}(\hilbert_2)\,.
\end{equation}

\medskip 

With these operational and algebraic assumptions in place, we can finally consider the bipartite model. Each mass is prepared in a spatial superposition of two orthogonal localized states,
\begin{equation}
\ket{\psi_i}=\tfrac{1}{\sqrt{2}}\big(\ket{L}_i+\ket{R}_i\big), \qquad \braket{L_i}{R_i}=0, \quad \quad i\in\{1,2\}\,,
\end{equation}
separated by a distance \(\Delta x\). The total initial state therefore reads
\begin{equation}
\label{eq:initial}
\ket{\psi(0)}_{12}
= \tfrac{1}{2}\big(\ket{L}_1+\ket{R}_1\big)\big(\ket{L}_2+\ket{R}_2\big)
\in \hilbert_1\otimes\hilbert_2\,.
\end{equation}

The only interaction considered between the two masses is their mutual Newtonian attraction,
\begin{equation}
\hat H_G = -\frac{G\,  m_1 m_2}{|\hat{\mathbf{r}}_1 - \hat{\mathbf{r}}_2|}\,.
\end{equation}
After an interaction time \(\tau\), the state becomes
\begin{equation}
\ket{\psi(\tau)}_{12} = \frac{e^{i\phi}}{\sqrt2}
\left\{
\ket{L}_1\,\frac{\ket{L}_2+e^{i\Delta\phi_{LR}}\ket{R}_2}{\sqrt{2}}
+
\ket{R}_1\,\frac{e^{i\Delta\phi_{RL}}\ket{L}_2+\ket{R}_2}{\sqrt{2}}
\right\}\,,
\end{equation}
with
\begin{equation}
\phi_{RL}\sim\frac{G m_1 m_2 \tau}{\hbar (d-\Delta x)},\quad
\phi_{LR}\sim\frac{G m_1 m_2 \tau}{\hbar (d+\Delta x)},\quad
\phi\sim\frac{G m_1 m_2 \tau}{\hbar d}\,,
\end{equation}
and \(\Delta\phi_{RL}=\phi_{RL}-\phi\), \(\Delta\phi_{LR}=\phi_{LR}-\phi\).
Whenever \(\Delta\phi_{LR}+\Delta\phi_{RL}\neq 2n\pi\) (\(n\in\mathbb{N}\)), the final state is entangled\footnote{To avoid dealing with entanglement between spatial degrees of freedom, which is difficult to measure, Bose \emph{et al.}'s setup includes spin degrees of freedom associated with the spatial positions and makes use of a Stern--Gerlach interferometer. Finally, entanglement is detected using an entanglement witness. We will not delve into these experimental details here, as our focus lies on the conceptual assumptions required for the argument's validity rather than on the implementation itself.}.

At this level of description, gravity is represented by a direct two-body interaction \(H_{12}\) on \(\Hcal_1 \otimes \Hcal_2\). Such a Hamiltonian is \emph{nonlocal} with respect to the bipartition because it cannot be written as a sum of strictly local terms \(H_1 \otimes \idm + \idm \otimes H_2\)\footnote{The locality notion here is not spatiotemporal. In finite-dimensional quantum mechanics, one distinguishes Bell nonlocality (instantiated as entangled states) from dynamical \(k\)-locality, where a \(k\)-local Hamiltonian acts nontrivially on at most \(k\) tensor factors \cite{Cotler:2017abq}. A bipartite interaction term \(H_{12}\) is \(2\)-local in this latter sense, but it is nonlocal because it jointly acts on both subsystems, which defines the whole system.}. Entanglement generated by such a bipartite interaction does not, by itself, imply that gravity possesses quantum degrees of freedom. The LOCC no-entanglement theorems apply only when the mediator acts \emph{locally} on each party and communicates classical information; they do not constrain scenarios in which a direct coupling is already present.

Therefore, the inference that gravity is quantum arises only once the 
interaction is reformulated as two couplings, \(H_{1G}\) and \(H_{2G}\), between each mass and a third subsystem \(G\) representing the gravitational mediator. In this tripartite picture, entanglement between \(m_1\) and \(m_2\) can emerge from initially factorized states only if \(G\) features a noncommutative operator algebra. We return to this point below.

\subsection{Tripartite description: gravity as a mediating subsystem}


In this \emph{tripartite} description, the total system is described by the Hilbert space
\begin{equation}
\hilbert = \hilbert_1 \otimes \hilbert_2 \otimes \hilbert_G\,,
\end{equation}
with corresponding observable algebras
\(\mathcal{A}_1, \mathcal{A}_2, \mathcal{A}_G \subset \mathcal{L}(\hilbert)\)\,,
each acting nontrivially only on its own factor. 
The initial global state is taken to be fully factorized,
\begin{equation}
\rho(0) = \rho_1 \otimes \rho_2 \otimes \rho_G\,,
\end{equation}
where the gravitational subsystem is prepared in a ``ready'' or background state 
\(\rho_G = \ket{\gamma_0}\!\bra{\gamma_0}\),
representing the unperturbed gravitational field prior to its interaction with the two masses. 

The key structural difference with respect to the bipartite case is that the Hamiltonian now consists of two couplings between each mass and gravity,
\begin{equation}
\hat H = \hat H_{1G} + \hat H_{2G}\,,
\end{equation}
with no direct term coupling the two masses. Each $\hat H_{iG}$ acts nontrivially only on $\hilbert_i \otimes \hilbert_G$. 
Since both couplings act on the gravitational subsystem, they do not commute in general, $[\hat H_{1G},\hat H_{2G}]\neq 0$, unless additional assumptions are made about the structure of the interaction on $\hilbert_G$. Nevertheless, in the weak-coupling and short-time regime relevant for the BMV protocol, 
the total time evolution can be approximated as
\begin{equation}
U(t)=e^{-i(\hat H_{1G}+\hat H_{2G})t/\hbar}
\simeq e^{-i\hat H_{1G}t/\hbar}\,e^{-i\hat H_{2G}t/\hbar}
\equiv U_{1G}(t)\,U_{2G}(t),
\end{equation}
up to higher-order corrections. This factorization implements the interactions between each mass and the gravitational subsystem as effectively local operations.

During the interaction, the gravitational state evolves conditionally on the spatial configuration of the two masses, coherently ``recording'' their branch information. In the path basis $X,Y\in\{L,R\}$, it is convenient to idealize the couplings as controlled unitaries on the mediator,
\begin{equation}
U_{1G}=\sum_{X\in\{L,R\}} \ket{X}\!\bra{X}_1\otimes I_2\otimes U_X,
\qquad
U_{2G}=\sum_{Y\in\{L,R\}} I_1\otimes \ket{Y}\!\bra{Y}_2\otimes V_Y,
\end{equation}
so that (with the appropriate time ordering if $[U_X,V_Y]\neq 0$)
\begin{equation}
U_{1G}U_{2G}\,\ket{X}_1\ket{Y}_2\ket{\gamma_0}
= \ket{X}_1\ket{Y}_2\,U_XV_Y\ket{\gamma_0}
\equiv e^{i\phi_{XY}}\ket{X}_1\ket{Y}_2\ket{\gamma_{XY}}\,,
\end{equation}
where the phases $\phi_{XY}$ reproduce those of the bipartite model, and each $\ket{\gamma_{XY}}$ represents the gravitational field state correlated with the corresponding configuration $(X,Y)$ of the masses. After an interaction time~$\tau$, the total state becomes
\begin{equation}
\ket{\Psi(\tau)}_{12G}
   = \frac{e^{i\phi}}{\sqrt2}
   \!\left[
   \ket{L}_1\frac{\ket{\gamma_{LL}}\ket{L}_2+e^{i\Delta\phi_{LR}}\ket{\gamma_{LR}}\ket{R}_2}{\sqrt{2}}
   + 
   \ket{R}_1\frac{e^{i\Delta\phi_{RL}}\ket{\gamma_{RL}}\ket{L}_2+\ket{\gamma_{RR}}\ket{R}_2}{\sqrt{2}}
   \right]\!\,.
\end{equation}

This tripartite model makes the locality structure explicit: the masses never interact directly but only through~$G$. 
As it's assumed that each mass couples to the gravitational field, the interaction Hamiltonians $\hat{H}_{1G}$ and $\hat{H}_{2G}$ correlate the field with the spatial branches of the masses. The gravitational subsystem thus evolves into a \emph{superposition of classical field configurations}, reflecting the distinct mass distributions, and in this sense exhibits a quantum character.  Under these conditions, the observation of entanglement between the masses implies that $G$ cannot be a classical system. 

The tripartite formulation therefore clarifies what is at stake in the BMV proposals. 
Entanglement between distant masses does not, by itself, demonstrate that gravity is quantum; but within a \emph{local-mediator} framework, it \emph{requires} gravity to be quantum. 
This perspective bridges the phenomenological bipartite model and the field-theoretic description, where the mediator’s degrees of freedom correspond to those of the gravitational field itself.

\vspace{1em}
\noindent
\textbf{Relation to the bipartite assumptions.}
The tripartite picture preserves the same \emph{operational} content as the bipartite one, while treating gravity as a third party. Concretely, the operational assumptions of \emph{independent preparation} and \emph{independent measurability} are now required to hold for each of the three parties (cf.\ Assumptions~\ref{aspt:prep-ind} and~\ref{aspt:comm-op}). Within the algebraic description, these operational requirements are then \emph{implemented} by the corresponding state-space and algebraic conditions on the chosen subalgebras (in particular, the existence of compatible global extensions and, when invoked, commutativity). The additional finite-dimensionality/completeness assumption (Assumption~\ref{aspt:finite}) remains purely structural and is what underwrites the use of an explicit tensor-product decomposition for the tripartite Hilbert space. Finally, the \emph{uncorrelated initial preparation} assumption (Assumption~\ref{aspt:uncorr}) is what gives the ``mediator'' inference its bite: if the initial global state carries no correlations between the two masses, then any later entanglement detected between them must be generated dynamically. 

\subsection{Field-theoretic description}\label{sec:BMV_field}

The tripartite description, in which gravity appears as a local mediating subsystem, is most naturally understood as an effective truncation of an underlying field-theoretic description, where the mediator’s degrees of freedom are those of the gravitational field itself. This perspective closely follows the strategy originally adopted in the BMV proposals~\cite{Bose,Bose_Supplementary_2017}, and has since become the standard framework for embedding the quantum-information-theoretic model \cite{Huggett:2022uui,Fragkos:2022tbm}. In this approach, the local couplings $\hat H_{iG}$ of the tripartite model arise as approximations to the standard matter--gravity interaction derived from linearized general relativity.

We start from a background spacetime with metric \(g_{\mu\nu}=\eta_{\mu\nu}+ \kappa h_{\mu\nu}\), where $\kappa^2 = 32\pi G,$ \(\eta_{\mu\nu}\) is the Minkowski metric and \(h_{\mu\nu}\) a small perturbation, subsequently quantized. Expanding the Einstein--Hilbert action to quadratic order and adopting the de~Donder gauge leads to the free-field Hamiltonian~\cite{Fragkos:2022tbm}
\begin{equation}
\hat H_G=\frac{1}{2}\!\int\! d^3k\,\hbar\omega_k
\left[
\hat b^{\dagger}_{\mu\nu}(\vec k)\hat b^{\mu\nu}(\vec k)
-\frac{1}{2}\hat b^{\dagger\mu}{}_{\mu}(\vec k)\hat b^{\nu}{}_{\nu}(\vec k)
\right],
\end{equation}
with \(\omega_k=c|\vec k|\) and \(\hat b_{\mu\nu}(\vec k)\), \(\hat b^\dagger_{\mu\nu}(\vec k)\) the graviton annihilation and creation operators. The coupling to matter follows from the usual interaction term
\begin{equation}
\hat H_{\text{int}}
=-\frac{\kappa}{2}\!\int\! d^3r\,\hat h^{\mu\nu}(\vec r)\,\hat T_{\mu\nu}(\vec r)\,,
\end{equation}
where \(\hat T_{\mu\nu}\) is the stress--energy tensor of the matter fields. In the nonrelativistic limit only the \(T_{00}\) component contributes,
\begin{equation}
\hat H_{\text{int}}
\simeq
-\frac{\kappa}{2}\!\int\! d^3r\,\hat h^{00}(\vec r)\,\hat T_{00}(\vec r)\,.
\end{equation}

To describe the two test masses, we model them as excitations of nonrelativistic scalar fields, with localized creation and annihilation operators, \(a^\dagger_{i,\xi}\) and \(a_{i,\xi}\), associated with the left and right branches of each interferometer, \(\xi\in\{L,R\}\). Introducing a finite quantization volume \(V\), so that \(\vec k\) takes discrete values\footnote{The discrete sum over momenta arises from quantizing the field in a finite auxiliary volume \(V\) with periodic boundary conditions. This is a standard normalization device: in the physical limit \(V\!\to\!\infty\), the sum becomes \(\sum_{\vec k}\to \tfrac{V}{(2\pi)^3}\int d^3k\), recovering the continuum of gravitational modes required for the Newtonian \(1/r\) potential.}, the volume-normalized Hamiltonian then reads~\cite{Bose_2022}

\begin{equation}
H =
\sum_{i,\xi} m_i c^2\, a^\dagger_{i,\xi} a_{i,\xi}
+ 
\sum_{\vec k} 
\hbar\omega_k\, b^\dagger_{\vec k} b_{\vec k}
-\hbar
\sum_{i,\vec k,\xi}
g_{i,\vec k}\,
a^\dagger_{i,\xi} a_{i,\xi}
\!\left(
b_{\vec k}\,e^{i\vec k\cdot \vec r_{i,\xi}}
+
b^\dagger_{\vec k}\,e^{-i\vec k\cdot \vec r_{i,\xi}}
\right),
\label{eq:ham-field}
\end{equation}
where $b^\dagger_{\vec{k}}$ creates an excitation of an \emph{effective scalar mediator mode} capturing the relevant Newtonian ($h_{00}$) sector in the weak-field, nonrelativistic regime, and $
g_{i,\vec k}= m_i c^2\sqrt{\frac{2\pi G}{\hbar c^3 k V}}\,.$
This Hamiltonian is the standard starting point used in field-theoretic embeddings of the BMV phase generation: it should be understood as an effective truncation of linearized gravity in which the mediator degrees of freedom retained are those that reproduce the static $1/r$ potential between nonrelativistic sources (rather than the radiative, transverse--traceless graviton polarizations). 

Once again, the total initial state is assumed to be fully uncorrelated,
\begin{equation}
\rho(0)=\rho_1\otimes\rho_2\otimes\rho_G\,,
\end{equation}
with $\rho_G=\ket{\gamma_0}\!\bra{\gamma_0}$ representing the gravitational vacuum or background state prior to its interaction with the two masses. After an interaction time $\tau$, the joint state of matter and gravity can be written schematically as
\begin{equation}
\ket{\Psi(\tau)} 
   = \frac{1}{2}\!
     \sum_{\xi,\xi'}
     \ket{\xi}_1\ket{\xi'}_2
     \!\bigotimes_{\vec k}\!
     e^{\,i\Phi_{\vec k,\xi\xi'}(\tau)} 
     \ket{\alpha_{\vec k,\xi,\xi'}}_G\,,
\end{equation}
where each $\ket{\alpha_{\vec k,\xi,\xi'}}$ is a coherent state of the gravitational field mode~$\vec{k}$ and the accumulated phase $\Phi_{\vec k,\xi\xi'}(\tau)$. The product over all $\vec k$ induces an effective entangling phase between the two matter branches that reproduces the gravitationally induced phase shifts $\Delta \phi_{\xi,\xi'}$ of the bipartite description~\cite{Bose_Supplementary_2017}.

\vspace{1em}
\noindent
\textbf{Relation to the bipartite and tripartite assumptions.}
At a formal level, the field-theoretic description seems to inherit the same \emph{operational} assumptions introduced in the bipartite and tripartite models: that the subsystems can be prepared independently (Assumption~\ref{aspt:prep-ind}), that the relevant observables are jointly measurable in the required sense (Assumption~\ref{aspt:comm-op}), and that the initial global state can be taken uncorrelated (Assumption~\ref{aspt:uncorr}). In the algebraic framework, these operational requirements are typically \emph{implemented} by suitable choices of subsystem algebras and state space. In addition, the use of an explicit Hilbert-space tensor product for matter and gravity implicitly relies on further structural input of the kind captured by Assumption~\ref{aspt:finite}.

However, unlike in the finite-dimensional quantum-mechanical setting, none of these assumptions is guaranteed once gravity is treated as a dynamical gauge field. In particular, the matter and gravitational algebras are infinite-dimensional, and gauge-invariant observables require gravitational dressing, which generically obstructs exact commutativity and hence the algebraic implementation for commensurability. As a result, while the field-theoretic construction above provides a controlled and widely adopted approximation scheme, the independence assumptions underlying the quantum-information-theoretic argument cannot be expected to hold exactly in the full theory. These issues will be analyzed in detail in the following sections.
\section{Notions of Independence}\label{section:notions_indep}

To assess the conceptual and mathematical underpinnings of the BMV proposals, and in particular Assumptions~\ref{aspt:prep-ind}--\ref{aspt:finite}, we now shift from a \emph{TPS-first} formulation to an \emph{algebra-first} one. This is motivated by two related considerations. First, the relevant gravitational interaction is most naturally treated as a (linearized) gauge field, for which gauge invariance is naturally encoded in terms of operator algebras rather than tensor factors. Second, the independence assumptions invoked in BMV-style arguments concern preparation and measurement capabilities, and in the context of field theory, should therefore be formulated without presupposing that a preferred tensor-product decomposition exists or is physically meaningful. For this reason, we adopt the language of algebraic quantum theory, and in particular algebraic quantum field theory  AQFT, where subsystem structure is represented by nets of von Neumann algebras and where one can cleanly separate \emph{operational} requirements (e.g.\ local preparability, no-signalling, independent preparations) from the \emph{algebraic} conditions often used to implement them (e.g.\ commutativity, split inclusions, and product-state
extension properties). The aim of this section is therefore not to introduce new physical assumptions, but to recast the independence notions already used in the proposal into a mathematically precise and more general form, and to make explicit the relations between operational requirements and their algebraic implementations.

Throughout this section, we describe physical systems in terms of algebras of observables and their states. Given a Hilbert space $\hilbert$, the algebra of all bounded operators is denoted $\mathcal{B}(\hilbert)$. The physically accessible observables of the system are described by a von Neumann (vN) algebra $\alga \subseteq \mathcal{B}(\hilbert)$, also called a $W^*$-algebra\footnote{All results discussed here admit counterparts for $C^*$-algebras. In this paper, we restrict attention to $W^*$-algebras. A $W^*$-algebra is a $C^*$-algebra that is closed in the weak operator topology (equivalently, in the strong operator topology for $*$-subalgebras of $\mathcal B(\mathcal H)$). In finite dimensions, the distinction between $C^*$- and $W^*$-algebras disappears \cite{redei2006quantumprobabilitytheory}.},
where \emph{observables} correspond to positive self-adjoint elements of $\alga$, while more general observables are constructed from these basic elements through functional calculus and algebraic combinations.
A state on $\alga$ is a positive, normalized linear functional $\omega:\alga \to \mathbb{C}$ assigning expectation values to observables\footnote{Since our focus lies on preparation and correlation properties, we restrict attention to states rather than more general completely positive maps.}. In the quantum-information setting of Section~\ref{sec:BMV}, states were represented by density operators on a Hilbert space; the algebraic notion adopted here generalizes this description and allows us to treat quantum-mechanical and field-theoretic systems on the same footing.

In nonrelativistic quantum mechanics, subsystem algebras are typically type~I factors\footnote{Factors are defined in Appendix \ref{appendix:extratheorems}.}. In this case, commutativity of subalgebras, together with finite-dimensionality and completeness assumptions, leads directly to a TPS of the Hilbert space \cite{Zanardi:2004zz}. By contrast, in AQFT the local algebras associated with spacetime regions are infinite-dimensional type~III vN algebras \cite{R_dei_2009}. As a result, subsystem independence can no longer be reduced to simple Hilbert-space factorization based on commutativity alone, and must instead be characterized directly at the algebraic level. 

We consider two subsystems $S_1$ and $S_2$ of a global physical system $S$, whose observables are described by a vN algebra $\alga \subseteq \mathcal{B}(\hilbert)$. We denote by $\alga_1,\alga_2 \subseteq \alga$ the observable algebras associated with $S_1$ and $S_2$. With this notation in place, we can now introduce precise algebraic notions of independence between subsystems, beginning with statistical independence.

\begin{definition}[Statistical independence \cite{summers1990}]
\label{def:Statistical independence}
A pair $(\alga_1,\alga_2)$ of vN subalgebras of a vN algebra $\alga$ is said to be statistically independent if, for every pair of normal states
$\omega_1:\alga_1 \to \mathbb{C}$ and $\omega_2:\alga_2 \to \mathbb{C}$,
there exists a normal state
$\omega:\alga_1 \vee \alga_2 \to \mathbb{C}$
such that
\begin{equation}
\omega|_{\alga_1} = \omega_1,
\qquad
\omega|_{\alga_2} = \omega_2\,.
\label{eq:statind}
\end{equation}
The state $\omega$ is called a \emph{normal joint extension} of $\omega_1$ and $\omega_2$.\footnote{Note that the restriction to normal states is essential here: one may define an analogous but weaker notion of independence for $C^*$-algebras, known as \emph{$C^*$-independence}, in which the states need not be normal \cite{redei2006quantumprobabilitytheory}. That is why statistical independence is referred to as $W^*$-independence in \cite{summers2009}. In this work, we focus exclusively on the $W^*$-algebraic setting. Extensions of this notion from states to completely positive maps lead to the concept of \emph{operational $W^*$-independence} \cite{summers1990,summers2009}, which we will not consider further.}
\end{definition}

Statistical independence is a \emph{state-space} (and hence algebraic) condition that can be \emph{read} as an implementation of the operational idea of independent local preparation: for any choice of local preparation procedures, represented by normal states on $\alga_1$ and $\alga_2$, there exists a global normal state on $\alga_1 \vee \alga_2$\footnote{The symbol $\alga_1 \vee \alga_2$ denotes the \emph{vN algebra generated} by $\alga_1$ and $\alga_2$, defined as $\alga_1 \vee \alga_2 := (\alga_1 \cup \alga_2)''$, where the double prime denotes the double commutant in $\mathcal{B}(\hilbert)$. By the double commutant theorem, $(\alga_1 \cup \alga_2)''$ is the smallest vN algebra containing both $\alga_1$ and $\alga_2$.} whose restrictions reproduce the given local statistics. Importantly, this condition does \emph{not} require the joint state to take a product form, nor does it preclude correlations between the subsystems. In the context of the BMV proposals, Assumption~\ref{aspt:prep-ind} is formulated operationally; when representing it within the algebraic framework, a natural sufficient implementation is precisely statistical independence in the sense of Definition~\ref{def:Statistical independence}.

The next notion of independence is \emph{statistical independence in the product sense}.

\begin{definition}[Statistical independence in the product sense \cite{summers1990}]
\label{def:Statistical independence in the product sense}
A pair $(\alga_1,\alga_2)$ of vN subalgebras of a vN algebra $\alga$ is said to be \emph{statistically independent in the product sense} if, for every pair of normal states
$\omega_1:\alga_1 \to \mathbb{C}$ and $\omega_2:\alga_2 \to \mathbb{C}$,
there exists a normal state
$\omega:\alga_1 \vee \alga_2 \to \mathbb{C}$
such that
\begin{equation}
\omega|_{\alga_1}=\omega_1,\qquad
\omega|_{\alga_2}=\omega_2,
\end{equation}
and
\begin{equation}
\omega(AB)=\omega_1(A)\,\omega_2(B)
\qquad
\forall\,A\in\alga_1,\; B\in\alga_2\,.
\end{equation}
The state $\omega$ is called a \emph{normal joint extension in product form}, or simply a \emph{normal product state}.\footnote{In the AQFT literature, this notion is referred to as \emph{$W^*$-independence in the product sense} \cite{summers1990,summers2009,redei2006quantumprobabilitytheory}. As in the previous case, the definition can be extended from states to completely positive maps, leading to the notion of \emph{operational $W^*$-independence in the product sense}.}
\end{definition}

Statistical independence in the product sense strengthens Definition~\ref{def:Statistical independence} by requiring not only mutual compatibility of arbitrary local preparations, but also the existence of a \emph{specific kind} of joint extension: an \emph{uncorrelated} one, in the precise sense that expectation values of joint observables factorize into a product of marginal expectations. This notion is therefore strictly stronger than statistical independence: it guarantees the existence of a joint extension of arbitrary local states \emph{and} imposes a constraint on the correlation structure of that extension. In the context of the BMV proposals, Assumption~\ref{aspt:uncorr} is formulated operationally as the possibility of an initial preparation with no correlations between the parties; when represented in the algebraic framework, a natural sufficient implementation is precisely statistical independence in the product sense.

Finally, as we already discussed, independence between subsystems can also be formulated at the level of measurements. This leads to the third notion of independence relevant for the BMV setting.

\begin{definition}[Commutativity]
\label{def:Commutativity}
Two subalgebras $\alga_1,\alga_2 \subseteq \alga$ are said to \emph{commute} if
\[
[A,B]=0
\qquad
\forall\,A\in\alga_1,\; B\in\alga_2\,.
\]
\end{definition}

Operationally, ``commensurability'' refers to the possibility of performing the relevant measurements jointly without mutual disturbance. In the algebraic framework, Definition~\ref{def:Commutativity} is a standard sufficient condition implementing this operational requirement. In the BMV proposals, this corresponds to Assumption~\ref{aspt:comm-op}.

\vspace{0.5em}
\noindent
State-preparation independence, absence of initial correlations, and measurement independence are three distinct aspects of subsystem independence that are all implicitly invoked in the BMV protocol. While these notions are closely implemented in standard finite-dimensional quantum mechanics (see below), they are not logically equivalent \cite{summers1990,summers2009}, and neither are the algebraic conditions commonly used to represent them. Specifically, statistical independence in the product sense is strictly stronger than statistical independence, and neither algebraic notion of preparation independence presupposes commutativity of the associated algebras\footnote{Although Definitions~\ref{def:Statistical independence} and \ref{def:Statistical independence in the product sense} do not assume commutativity, in many cases of interest, notably for factors/type~III algebras, the existence of normal product-state extensions is closely tied to split inclusions, which in turn entail $\alga_1\subseteq \alga_2'$., as discussed in Sec. \ref{ind_field_th}.}.

\subsection{Back to the Bipartite and Tripartite Models}

In finite-dimensional nonrelativistic quantum mechanics, the distinctions introduced above effectively disappear. Subsystem algebras are typically \textit{finite} type~I factors, and commutativity together with completeness assumptions suffices to induce a TPS of the global Hilbert space, as discussed in Sec.~\ref{sec:BMV}. Let's formally define it:

\begin{definition}[Tensor product structure \cite{Cotler:2017abq}] \label{def_TPS}
A \textit{TPS} of $\mathcal{H}$ is an equivalence class of isomorphisms $T : \mathcal{H} \rightarrow \bigotimes_{i=1}^n \mathcal{H}_i$ that factorize $\mathcal{H}$ into $n$ factors $\mathcal{H}_i$ of respective dimensions $d_i$, where two isomorphisms $T_1$ and $T_2$ are said to be equivalent (denoted $T_1 \sim T_2$) if $T_1 T_2^{-1}$ is a product of local unitaries $U_1 \otimes \dots \otimes U_n$ and arbitrary permutations of the factors.
\end{definition}
\noindent 
Once a TPS is available, the above operational notions of independence are promptly implemented. 

However, this conclusion relies essentially on finite dimensionality. Already for infinite-dimensional type~$\mathrm{I}_\infty$ factors, commutativity and completeness conditions no longer suffice to determine a TPS without additional assumptions\footnote{%
This reflects a genuine loss of rigidity in infinite dimensions. In the finite-dimensional case, a collection of mutually commuting type~I factors that jointly generate $\mathcal B(\mathcal H)$ uniquely fixes, up to unitary equivalence, a TPS in $\mathcal H$ \cite{Zanardi:2004zz}. By contrast, for type~$\mathrm{I}_\infty$ factors, the same algebraic data admit infinitely many inequivalent TPSs: an inclusion $\mathcal A\cong\mathcal B(\mathcal K)\subset\mathcal B(\mathcal H)$ determines $\mathcal H$ only up to a non-unique infinite-multiplicity decomposition $\mathcal H\cong\mathcal K\otimes\mathcal K'$. As a result, commuting subalgebras that generate the full operator algebra do not, by themselves, single out a canonical TPS \cite{KadisonRingrose1986}.}.
This already marks a first departure from the quantum-informational intuition based on Hilbert-space factorization.

From the perspective of the bipartite and tripartite models used in the BMV analysis, this observation is important: while those descriptions tacitly rely on a well-defined TPS to formulate independent preparation and measurement, uncorrelated states, and local couplings to a mediating subsystem, such a structure is not guaranteed beyond the finite-dimensional setting.

\subsection{Back to the field-theoretic description}\label{ind_field_th}

We now turn to the field-theoretic setting, where the subsystem structure underlying the BMV proposals is most naturally analyzed within AQFT. In this framework, physical systems are described by a \emph{net} of vN algebras $\{\alga(\mathcal O)\}$ indexed by spacetime regions $\mathcal O$, rather than by a fixed choice of abstract subalgebras. The interpretation of subsystems is therefore intrinsically (spatiotemporally) local.

For the purposes of the present discussion, the relevant AQFT axioms are:
(i) \emph{isotony}, which ensures that if $\mathcal O_1\subset\mathcal O_2$ implies the algebraic inclusion $\alga(\mathcal O_1)\subset\alga(\mathcal O_2)$; and
(ii) \emph{microcausality}, according to which algebras associated with spacelike separated regions commute,
\[
[\alga(\mathcal O_1),\alga(\mathcal O_2)]=0,
\qquad \text{for } \mathcal O_1 \perp \mathcal O_2 .
\]
In this setting, subalgebra inclusions and commutativity are no longer an independent assumption imposed on a pair of chosen subalgebras, but a structural consequence of spacetime geometry encoded in the net itself. 



As discussed, the local algebras $\alga(\mathcal O)$ are generically type~III, so even for spacelike separated regions there is typically no canonical Hilbert-space TPS and normal product-state extensions of arbitrary local preparations are not automatic. Recovering (even approximately) the tensor-product intuition underlying the bipartite and tripartite BMV models therefore requires additional algebraic input. A key role is played by intermediate type~I factors, which reinstate many familiar finite-dimensional constructions ``locally''. This motivates introducing the \emph{funnel property}, which precisely encodes the existence of such intermediate type~I structure in the AQFT setting.

\begin{definition}[Funnel Property \cite{summers2009}] \label{def:funnel}
The inclusion $\alga(U)\subset\alga(W)$ is said to have the funnel property if there exists a type~I factor $\algm$ such that
\[
\alga(U)\subset \algm \subset \alga(W).
\]
\end{definition}

Physically, the funnel property expresses the existence of a finite ``buffer'' or ``collar'' region between $U$ and $W$ that allows the isolation of degrees of freedom localized in $U$, in a way reminiscent of a tensor product decomposition, despite the underlying local algebras being of type~III.

When combined with microcausality, the funnel property has direct consequences for bipartite systems. Concretely, consider a second region $V$ that is spacelike separated from $W$ (and hence also from $U$). By microcausality, this geometric separation is equivalent to the algebraic commutation relation $[\alga(W),\alga(V)]=0$, which can be rewritten as $\alga(W)\subset \alga'(V)$. Then, if the inclusion $\alga(U)\subset\alga(W)$ has the funnel property, then
\[
\alga(U)\subset \algm \subset \alga'(V).
\]
\noindent 
This is an instance of the \textit{split property} for the pair $(\alga(U),\alga(V))$ in the case of local subalgebras associated with spacetime regions. However, it holds more generally:
\begin{definition} (Split property)
    A pair $(\alga_1,\alga_2)$ of von Neumann algebras is split if there exists a type~I factor $\algm$ such that $\alga_1\subset\algm\subset\alga_2'$\footnote{In \cite{doplicher_standard_1984}, the notation is for the pair $(\alga_1,\alga_2')$ instead, but we follow \cite{summers1990,summers2009,R_dei_2009}.}. 
    \label{def:split property}
\end{definition}

The split property is not only connected to commutativity, but also with statistical independence in the product sense, and the existence of a Hilbert-space TPS, as made precise by the following theorems.

\begin{theorem}(\cite[Thm.~5.2]{summers2009})
If the pair of vN algebras $(\alga_1,\alga_2)$ is split, then the pair is also statistically independent in the product sense. Moreover, if either of the algebras is type~$III$ or a factor, then the pair is statistically independent in the product sense if and only if it is split.\footnote{This theorem was originally given for operational independence, but the notion of operational independence discussed here was proven to be equivalent to statistical independence in the product sense \cite{R_dei_2009}.}
\label{th:ifsplitthenopidpdt}
\end{theorem}

\begin{theorem}(\cite[Thm.~4.1]{summers2009})
For a mutually commuting pair $(\alga_1,\alga_2)$ of von Neumann algebras acting on a Hilbert space $\hilbert$, the following are equivalent.
\begin{enumerate}\itemsep0pt\parskip0pt\parsep0pt\topsep0pt\partopsep0pt 
    \item The pair $(\alga_1,\alga_2)$ is split;
    \item The pair $(\alga_1,\alga_2)$ is $W^*$-independent in the spatial product sense, i.e.\ the map
    \[
    AB \longmapsto A\otimes B, \qquad A\in\alga_1,\; B\in\alga_2,
    \]
    extends to a spatial isomorphism of $\alga_1\vee\alga_2$ with $\alga_1\bar{\otimes}\alga_2$\footnote{The $W^*$-tensor product between two von Neumann algebras $\alga_1$ and $\alga_2$ is denoted $\alga_1\bar{\otimes}\alga_2$, meaning that the operators contained in each subalgebra act on distinct factors of the Hilbert space $\hilbert$.}, i.e.\ there exists a unitary operator $U:\hilbert\rightarrow\hilbert\otimes\hilbert$ such that $UABU^*=A\otimes B$ for all $A\in\alga_1$ and $B\in\alga_2$.
\end{enumerate}
\label{th:equivalentpairsplitspatialisomorphism}
\end{theorem}

As $W^*$-statistical independence in the spatial product sense is strictly stronger than $W^*$-statistical independence in the product sense \cite{DANTONI1983361,R_dei_2009}, these theorems show that the split property plays a dual role:  it is the condition under which statistical independence in the product sense is guaranteed, and it implies a Hilbert-space factorization where each algebra can be embedded in a different tensor factor. Thus, in quantum field theory, statistical independence in the product sense cannot be assumed on the basis of microcausality alone. The relations between the different notions of independence introduced above are summarized in Figure~\ref{graph:small}.

Therefore, the split property acts as the central link between the different notions of independence. If a pair of vN algebras fails to commute, it cannot be split, and this obstructs having a TPS. In such cases, either normal product states do not exist, or the notion of subsystem separation cannot be implemented at the Hilbert-space level. By contrast, statistical independence without the product requirement is the weakest notion considered here and the only one that does not presuppose commutativity; it may therefore still hold even in the absence of a split inclusion. We revisit this in the context of gravity in the Discussions. 

\tikzset{
  block/.style={
    rectangle, draw, rounded corners=2pt,
    fill=gray!10,
    align=center,
    font=\small,
    text width=10.8em,
    minimum height=3.7em,
    inner sep=3pt,
    drop shadow={opacity=.85, shadow xshift=1.3pt, shadow yshift=-1.3pt}
  },
  rel/.style={-Stealth, thick},
  rel2/.style={Stealth-Stealth, thick},
  lab/.style={
    font=\footnotesize,
    fill=white,
    inner sep=1.2pt,
    fill opacity=0.9,
    text opacity=1
  }
}

\begin{figure}[H]
  \centering
  \begin{tikzpicture}[node distance = 3cm, auto]
 \node [block] (split) {Split property \\ (Def.  \ref{def:split property})};
     \node [block, left of = split, node distance = 6cm] (structureH) {Tensor Product Structure \\(Def. \ref{def_TPS})};
     \node [block, right of = split, node distance = 6cm] (commensurability) {Commutativity \\ (Def.  \ref{def:Commutativity})};
     \node[block] (funnel) at ($(split)!0.5!(commensurability) + (0,3cm)$)
{Funnel Property \\ (Def. \ref{def:funnel})};
    \node [block, below of = split, node distance = 3cm] (statidpdproduct) {Statistical independence in the product sense \\ (Def. \ref{def:Statistical independence in the product sense})};
   \node [block, right of  = statidpdproduct, node distance = 6cm] (statidpdt) {Statistical independence \\ (Def.  \ref{def:Statistical independence})};
   
 \draw[rel2] (split) -- node[midway, above=2pt, sloped] {Thm. \ref{th:equivalentpairsplitspatialisomorphism}}(structureH);
    \draw [rel] (split) --(commensurability);
    \draw [rel] (split) --node[midway, left] {Thm. \ref{th:ifsplitthenopidpdt}} (statidpdproduct);
    \draw [rel2]([xshift=0.625cm]split.south) --node[midway, right] {If type III/factor} ([xshift=0.625cm]statidpdproduct.north);
      \draw [rel] (statidpdproduct) -- (statidpdt);
\coordinate (fjoin) at ($(funnel.south)+(0,-0.8cm)$);
\draw[thick] (funnel.south) -- (fjoin);
\draw[thick] (fjoin) -| (commensurability.north);
\draw[rel] (fjoin) -| (split.north);
  \end{tikzpicture}
  \caption{Summary of the properties introduced in this section, and their respective relations.}
  \label{graph:small}
\end{figure}
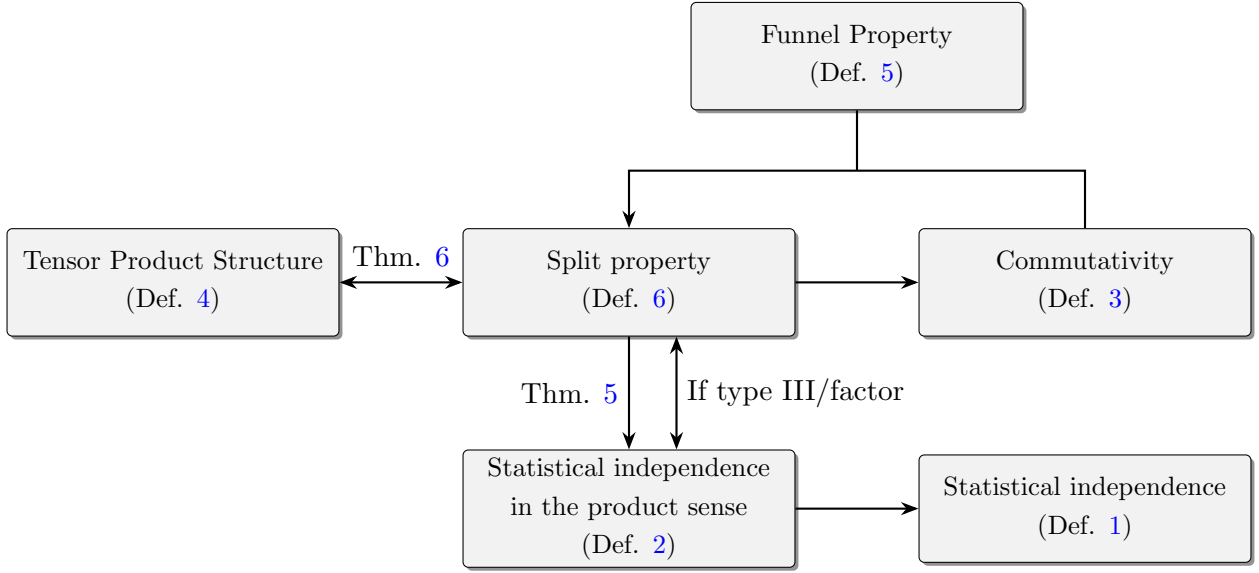

\paragraph{Are the required algebraic conditions obtained in linearized quantum gravity?}
This section shows that, within a field-theoretic description, statistical independence in the product sense can indeed be recovered—but only under rather stringent algebraic conditions. In particular, it is the joint assumption of microcausality and the funnel property that allows one to reconstruct the independence assumptions tacitly invoked in the BMV proposals, including the existence of uncorrelated product states and a TPS for spacelike separated subsystems.

However, this entire construction hinges on a crucial premise: that the physical subsystems under consideration admit a description in terms of local vN algebras that \textit{actually} satisfy these properties. Thus, before one can meaningfully appeal to the split property to justify preparation and measurement independence, one must ask whether the relevant (gauge-invariant) gravitational observables even generate commuting local algebras to begin with, and whether such algebras can admit split inclusions at all. In other words, do the algebraic conditions required to get the BMV logic off the ground survive once the gauge nature of gravity is taken seriously?

The next section addresses precisely this issue. By examining the role of gravitational dressing and gauge constraints, we will question whether the local subalgebras implicitly assumed in Sec. \ref{sec:BMV_field} can be consistently defined, and whether microcausality—and with it the split property—can be expected to hold, even approximately, for physically meaningful observables in quantum gravity.

\section{Dressed states and Commutativity}\label{section:Dressed states and Commutativity}

In quantum field theory, physical observables are represented by gauge-invariant Hermitian operators acting on a Hilbert space. In gauge theories, local matter fields are typically not gauge invariant and therefore do not define observables by themselves. A standard procedure to construct gauge-invariant operators is to \emph{dress} such fields \cite{Kibble:1968sfb,Kulish:1970ut,DonnellyFreidel2016LocalSubsystems,Giddings2019SoftChargesSplitting,Giddings_2006,Bagan_2000}. While dressing restores gauge invariance, it generically modifies the localization properties of the resulting operators and, in particular, their commutation relations. This is precisely where the algebraic implementation of commensurability via commutativity might be obstructed. As we will see, this can be avoided in Yang-Mills theories, but it seems to become unavoidable in gravity. We review the situation in QED and gravity following \cite{Donnelly:2015hta,Donnelly:2016rvo}.

\subsection{QED}

Consider QED coupled to a complex scalar field $\phi$ of charge $q$,
\begin{equation}
    \mathcal{L}_{\mathrm{QED}}
    =
    -\frac{1}{4}F^{\mu\nu}F_{\mu\nu}
    -\frac{1}{2\alpha}(\partial_\mu A^\mu)^2
    -|D_\mu\phi|^2
    -m^2|\phi|^2 ,
\end{equation}
with $D_\mu=\partial_\mu-iqA_\mu$, and $\alpha$ the gauge fixing parameter. The theory is invariant under the gauge transformations
\begin{equation}
    A_\mu(x)\rightarrow A_\mu(x)-\partial_\mu\Lambda(x),
    \qquad
    \phi(x)\rightarrow e^{-iq\Lambda(x)}\phi(x),
\end{equation}
where we assume $\Lambda(x)\to 0$ at spatial infinity.

The matter field $\phi$ is not gauge invariant and therefore does not define a physical observable in the interacting theory. Gauge-invariant operators can be constructed by dressing the field,
\begin{equation}
    \Phi(x)=V(x)\phi(x),
\end{equation}
where the dressing $V(x)$ transforms oppositely to $\phi(x)$. A convenient class of dressings is given by
\begin{equation}
    V(x)=\exp\!\left(
    iq\!\int\! d^4x'\, f^\mu(x,x')A_\mu(x')
    \right),
\end{equation}
with $\partial'_\mu f^\mu(x,x')=\delta^{(4)}(x-x')$. Different choices of $f^\mu$ correspond to different physical dressings (Coulomb, line, or Lorenz dressings), and describe different photon clouds accompanying the charged excitation.

Although $\Phi(x)$ is gauge invariant, it is not strictly local: typical dressings extend to infinity.
Accordingly, gauge-invariant \emph{charged} operators generally do \emph{not} satisfy microcausality in the strict AQFT sense.
For example, for Dirac/Coulomb-type dressings one finds nontrivial commutators once one considers unequal-time commutators (equivalently, $\,[\dot\Phi,\Phi]$ at equal time), with the leading behavior controlled by the Coulomb interaction energy \cite{Donnelly:2015hta,Donnelly:2016rvo}.

By contrast, genuinely local gauge-invariant operators with compact support (e.g.\ Wilson loops, or neutral ``mesonic'' operators with both endpoints contained in a bounded region) \emph{can} be associated with commuting subalgebras for spacelike separated supports. This result follows from the fact that the dressing is constructed from the gauge field $A_\mu$, which is itself the canonically quantized variable. Consequently, at equal times, all contributions to the commutator vanish. In this sense, QED/Yang--Mills admits local observable algebras obeying microcausality, but \emph{single-charge} dressed fields are intrinsically nonlocal and need not commute at spacelike separation.


\subsection{Gravity}

We now turn to gravity coupled to a scalar field $\phi$ of mass $m$,
\begin{equation}
    \mathcal{L}
    =
    \frac{2}{\kappa^2}R
    -\frac{1}{\alpha\kappa^2}\frac{\sqrt{|g^0|}}{\sqrt{|g|}}\frac{1}{|g|}
    \left[\nabla^0_\mu\!\left(\sqrt{|g|}g^{\mu\nu}\right)\right]^2
    -\frac{1}{2}\big[(\nabla\phi)^2+m^2\phi^2\big],
\end{equation}
where $\kappa^2=32\pi G$, $g_{\mu\nu}$ is the spacetime metric, and $\alpha$ the gauge fixing parameter. Working in perturbative quantum gravity, we expand around a fixed Minkowski background,
\begin{equation}
    g_{\mu\nu}=\eta_{\mu\nu}+\kappa h_{\mu\nu}.
\end{equation}

As in QED, the scalar field $\phi$ is not gauge invariant under diffeomorphisms and therefore does not by itself define an observable. Gauge-invariant operators can be constructed by \emph{gravitationally dressing} $\phi$, i.e.\ by defining a relationally localized field
\begin{equation}
    \Phi(x)=\phi(x)+V^\mu(x)\,\partial_\mu\phi(x)+\mathcal{O}(V^2),
\end{equation}
where $V^\mu(x)$ is chosen so that $\Phi(x)$ is invariant under linearized diffeomorphisms. At the linearized level one may formally write
\begin{equation}
    V^\mu(x)=\kappa\!\int\! d^4x'\, f^{\mu\nu\lambda}(x,x')\,h_{\nu\lambda}(x')\,,
    \qquad
    2\,\partial'_\nu f^{\mu\nu\lambda}(x,x')=\delta^{(4)}(x-x')\,\eta^{\mu\lambda}\,,
\end{equation}
which guarantees that $V^\mu$ transforms as the required spacetime displacement. However, unlike in QED, the kernel $f^{\mu\nu\lambda}$ cannot in general be taken to be strictly local: diffeomorphism-invariant dressings necessarily involve nonlocal functionals of the metric perturbation, reflecting the fact that gravitationally dressed operators must carry their associated long-range field.

A concrete illustration is provided by the (linearized) gravitational Wilson-line dressing \cite{Donnelly:2015hta},
\begin{equation}
V_\mu(x)=\frac{\kappa}{2}\int_{x}^{\infty} dx'^{\nu}\left[
h_{\mu\nu}(x')+2\int_{x'}^{\infty} dx''^{\lambda}\,\partial_{[\mu}h_{\nu]\lambda}(x'')
\right],
\end{equation}
where the integration is along an arbitrary curve from $x$ to spatial infinity. Physically, this dressing corresponds to a ``gravitational string'' extending to infinity. The crucial structural difference with Yang--Mills dressings is the appearance of derivatives of the metric perturbation in $V_\mu$. As a result, gravitationally dressed fields generically have nonvanishing equal-time commutators even for spacelike separated arguments \cite{Donnelly:2015hta,Donnelly:2016rvo,franzmann2024be}. Indeed, already at order $\kappa^2$, one finds for the Coulomb dressing (obtained by averaging the Wilson line over all directions) that \cite{Donnelly:2015hta}
\begin{equation}
    [\Phi(x),\Phi(x')]
    =
    -\frac{i\kappa^2\hbar}{64\pi c^4}
    \big[
    \dot\phi(x)\partial_i\phi(x')
    +
    \partial_i\phi(x)\dot\phi(x')
    \big]
    \frac{x^i-x'^i}{|x-x'|}.
\end{equation}
Similarly, the commutator between the field and its time derivative,
\begin{equation}
    [\dot\Phi(x),\Phi(x')]
    =
    -\frac{i\kappa^2\hbar}{32\pi c^4|x-x'|}
    \dot\phi(x)\dot\phi(x'),
\end{equation}
is nonvanishing at equal times. We emphasize that this holds for representative dressings: while some commutator structures can be modified by changing the dressing, in particular, the first commutator can be made to vanish, no dressing can make the second commutator vanish. 

The physical consequence is that the gauge-invariant algebras generated by dressed mass operators associated with spacelike separated regions fail to satisfy microcausality. As a result, the commensurability required to define independent subsystems through the split property fails, and the factorization of the Hilbert space, as well as the existence of uncorrelated states assumed in gravitationally induced entanglement experiments, becomes problematic. 

There is substantial evidence that genuinely gauge-invariant \emph{local} observables are generically unavailable in quantum gravity beyond the context of the linearized theory. In closed universes, diffeomorphism-invariant observables cannot be built as local functions of the canonical variables \cite{Torre:1993fq}. Even when classical gravitational observables can be arranged to be microcausal, this property is expected to fail once quantum effects are included \cite{Marolf:2015jha}. Moreover, operators carrying nonzero Poincar\'e charges must be gravitationally dressed, rendering the corresponding gauge-invariant operators intrinsically nonlocal and obstructing sharply localized subsystem definitions \cite{Donnelly:2016rvo,Giddings:2018cjc}. For further references, see \cite[section \textit{Observables} and references therein]{deBoer:2022zka} and \cite{EmerGe_proj_1}.

\subsection{Failure of the split property in linearized gravity}
\label{subsec:loss-einstein-separability}

The results of the previous section have consequences that go beyond the failure of microcausality for specific dressed observables. They directly affect the algebraic implementation required to model the two test masses as independent quantum subsystems in the sense presupposed by the BMV/GIE protocols. Importantly, the dressed operators we considered are not strictly localized in bounded regions (they are relationally defined with an asymptotic anchor). The point is therefore not that the abstract AQFT net axioms fail, but that the \emph{gauge-invariant} subalgebras naturally associated with these dressed degrees of freedom need not commute at spacelike separation, obstructing their use as commuting ``local'' subsystem algebras in the BMV/GIE sense.

As reviewed in Sec.~\ref{section:notions_indep}, the existence of a TPS for a bipartite (or tripartite) system is not a primitive notion in quantum field theory. In the algebraic setting, a TPS compatible with the observable content of the theory can be recovered if the relevant subsystem algebras satisfy the split property. Only then can one represent the global algebra as a tensor product in a Hilbert space, allowing for product states, and the standard quantum-information-theoretic description of subsystems.


The breakdown of microcausality for gravitationally dressed observables therefore has immediate consequences: they cannot form a split pair. As a result, no Hilbert-space tensor factorization exists that is compatible with the gauge-invariant observable algebras. In particular, the notion of a product state of the form $\rho_1 \otimes \rho_2$---or, equivalently, a normal product state extending arbitrary local preparations---seems to be obstructed at the fundamental level. Thus, once gravitational dressing is taken into account, subsystem independence becomes, at best, an approximate notion, valid only insofar as violations of commutativity can be neglected; this dovetails with the effective-gravity viewpoint that diffeomorphism-invariant observables are intrinsically nonlocal and reproduce local QFT only approximately in suitable states/limits (``pseudo-locality'') \cite{Giddings_2006}. We will come back to this point in Sec.~\ref{sec:microcausality_violation}.

\vspace{0.5em}
\noindent
\textbf{Consequences for Bell-type reasoning.}
Much of the quantum-information-theoretic interpretation of BMV/GIE-style experiments is phrased in terms of \emph{entanglement generation} between two parties and, more broadly, in terms of LOCC/Bell-style constraints that presuppose a clean bipartite subsystem structure. Since the split property and exact commutativity are precisely the ingredients that underwrite such a structure in the algebraic setting, it is natural to ask what becomes of Bell-type correlation bounds once these ingredients fail.

At this point, the standard Bell framework no longer strictly applies. Bell inequalities, and in particular the CHSH inequality, presuppose a bipartite structure with well-defined local observables acting on commuting subalgebras, as well as the existence of product states and independent local measurements. When these assumptions fail, Bell inequalities cannot be interpreted in their usual way as constraints on correlations between independent subsystems.

Nevertheless, one may still ask whether meaningful bounds on observable correlations survive once commutativity is relaxed. Rather than diagnosing nonlocality in the sense of local hidden variable theories, such generalized Bell-type inequalities probe the algebraic structure of correlations themselves, and the extent to which noncommutativity modifies their maximal strength. This motivates a careful re-examination of Bell inequalities—and in particular of Tsirelson’s bound—from the perspective of noncommuting observable algebras, which we undertake in the next section.

\section{Bell inequalities: the CHSH case}\label{sec:bell}

BMV-type experiments are \emph{not} designed as Bell tests: they aim to detect entanglement via suitable
\emph{entanglement witnesses}, i.e.\ observables whose expectation values certify nonseparability without
necessarily excluding local hidden-variable (LHV) models. Hence, entanglement witnesses are (in general)
strictly weaker than Bell inequalities: a Bell violation implies entanglement, but not conversely.

That said, Bell inequalities and entanglement witnesses are structurally close. In particular, every Bell
inequality can be associated with an entanglement witness \cite{Terhal_2000}, and both may be viewed as
constraints separating separable from entangled correlations. For our purposes, the CHSH inequality provides
a clean benchmark: it comes with a sharp classical bound (for LHV models) and a sharp quantum bound
(Tsirelson's bound \cite{cirelson_1980}).

The point of this section is therefore \emph{diagnostic}, not interpretive: we do not reinterpret BMV as a Bell
experiment. Rather, we ask which steps in the CHSH/Tsirelson analysis rely on the \emph{algebraic}
implementation of commensurability, namely commutativity ($[A_i,B_j]=0$), and what remains meaningful once
this assumption is relaxed, as suggested by gravitational dressing. This will later inform how robust standard entanglement-witness bounds are beyond the usual
tensor-product assumption.

\subsection{CHSH inequality and Tsirelson's bound under commensurability}

Consider two parties $A$ and $B$ with two observables each, $A_1,A_2$ and $B_1,B_2$
(typically taken self-adjoint with spectrum in $\{\pm1\}$, or more generally contractions with $\|A_i\|\le 1$,
$\|B_j\|\le 1$). If the observed correlations admit an LHV model \cite{PhysicsPhysiqueFizika.1.195}, then they
satisfy the Clauser--Horne--Shimony--Holt (CHSH) inequality \cite{1969PhRvL..23..880C},
\begin{equation}
    \big|\langle A_1B_1\rangle + \langle A_1B_2\rangle
    + \langle A_2B_1\rangle - \langle A_2B_2\rangle\big|
    \leq 2 .
\end{equation}
Quantum mechanics can violate this bound for suitable choices of measurements on entangled states, thereby
excluding LHV models (see e.g.\ \cite{1969PhRvL..23..880C,PhysRevD.14.1944,PhysRevLett.49.1804}).

On the quantum side, it is convenient to package the CHSH expression into the CHSH operator
\begin{equation}
    E := A_1B_1 + A_1B_2 + A_2B_1 - A_2B_2 .
    \label{eq:b}
\end{equation}
Under the standard subsystem assumptions, in particular, that $A$-observables commute with $B$-observables, 
$E$ is self-adjoint whenever the $A_i$ and $B_j$ are, and one can derive Tsirelson's bound
\begin{equation}
    \|E\| \le 2\sqrt{2}\,,
\end{equation}
equivalently $\langle E\rangle \le 2\sqrt{2}$ for all quantum states. This derivation uses commutativity in an
essential way (both algebraically and conceptually): it guarantees that joint observables such as $A_iB_j$ are
well-defined as observables of the composite system, and it matches the usual tensor-product picture (where
$A_i=A_i\otimes I$ and $B_j=I\otimes B_j$).

\subsection{Loss of commutativity and a symmetrized CHSH observable}

As discussed in Sec.~\ref{section:Dressed states and Commutativity}, gravitational dressing generically obstructs exact commutativity even for spacelike separated degrees of freedom. In that situation, the usual CHSH expression is no longer guaranteed to define a self-adjoint observable: if $A_i$ and $B_j$ are self-adjoint but do not commute, the products $A_iB_j$ need not be self-adjoint.

A minimal way to retain the CHSH structure while ensuring self-adjointness is to replace operator products by their Jordan (symmetrized) product,
\begin{equation}
    X\circ Y := \tfrac12(XY+YX)\,,
\end{equation}
which is self-adjoint whenever $X$ and $Y$ are\footnote{For bounded operators $A,B\in\mathcal B(\mathcal H)$, the Jordan product $A\circ B$ is self-adjoint whenever
$A$ and $B$ are self-adjoint, and reduces to the ordinary product when $[A,B]=0$.
Here it is used purely as a formal symmetrization to build self-adjoint expressions from noncommuting factors;
we do not assume any specific operational measurement protocol for $A\circ B$ in the generic noncommuting case.}.
This leads to the symmetrized CHSH observable
\begin{equation}
    E_\circ \;:=\; A_1\circ B_1 + A_1\circ B_2 + A_2\circ B_1 - A_2\circ B_2
    \;=\;\frac{E+E'}{2}\,,
    \label{eq:Esym}
\end{equation}
where $E'=B_1A_1+B_2A_1+B_1A_2-B_2A_2$ is the CHSH expression with reversed ordering. By construction, $E_\circ$ is self-adjoint for self-adjoint $A_i,B_j$, and reduces to the usual CHSH operator $E$ in the commuting case.

The key question is then whether the usual Tsirelson bound survives in this noncommuting setting. The following proposition shows that it does: once the CHSH expression is reformulated in terms of observables via the Jordan product, violations of commutativity do \emph{not} generate stronger-than-quantum correlations.

\begin{prop}\label{prop:jordan-tsirelson}
Let $A_1,A_2,B_1,B_2$ be self-adjoint observables satisfying $A_i^2=B_j^2=I$. Without assuming commensurability (i.e.\ without requiring $[A_i,B_j]=0$), the symmetrized CHSH operator
\begin{equation}
    E_\circ \;=\; A_1\circ B_1 + A_1\circ B_2 + A_2\circ B_1 - A_2\circ B_2
    \;=\;\frac{E+E'}{2}
\end{equation}
satisfies the bound
\begin{equation}
    E_\circ \;\leq\; 2\sqrt{2}\, I\,,
\end{equation}
where $E=A_1B_1+A_1B_2+A_2B_1-A_2B_2$ and $E'=B_1A_1+B_2A_1+B_1A_2-B_2A_2$.
\end{prop}

\begin{proof}
We rewrite $E_\circ$ as a difference of Jordan squares. A direct algebraic rearrangement yields
\begin{equation}
\begin{split}
    E_\circ
    &=A_1\circ B_1+A_1\circ B_2+A_2\circ B_1-A_2\circ B_2\\
    &=\frac{1}{\sqrt{2}}\Big(A_1\circ A_1+A_2\circ A_2+B_1\circ B_1+B_2\circ B_2\Big)
    -\frac{\sqrt2-1}{8}\sum_{k=1}^4 F_k\circ F_k\,,
\end{split}
\end{equation}
where $F_{1}=(\sqrt2+1)(A_1-B_1)+A_2-B_{2},\, F_{2}=(\sqrt2+1)(A_1-B_2)-A_{2}-B_{1},\,F_{3}=(\sqrt2+1)(A_2-B_1)+A_{1}+B_{2},\,F_{4}=(\sqrt2+1)(A_2+B_2)-A_{1}-B_{1}$ are sum of observables and therefore observables themselves.

For any self-adjoint $X$, one has $X\circ X=X^2\ge 0$. Using $A_i^2=B_j^2=I$, we therefore obtain
\begin{equation}
    E_\circ
    =\frac{1}{\sqrt{2}}\big(A_1^2+A_2^2+B_1^2+B_2^2\big)
    -\frac{\sqrt2-1}{8}\sum_{k=1}^4 F_k^2
    \;\le\;\frac{1}{\sqrt{2}}(4I)
    \;=\;2\sqrt2\,I\,.
\end{equation}
\end{proof}

Proposition~\ref{prop:jordan-tsirelson} shows that the Tsirelson limit is \emph{robust}: even when $[A_i,B_j]\neq 0$, the symmetrized CHSH observable $E_\circ$ cannot exceed $2\sqrt2$ in quantum theory. This is not a trivial recovery of the commuting case, however. When the subsystem algebras commute, $E_\circ=E$ and we recover the standard CHSH operator. When they do not commute, by contrast, neither ordering $A_iB_j$ nor $B_jA_i$ defines a joint observable in general; only their symmetrized combination does. In this sense, the Jordan product implements a formal ``average'' over incompatible orderings enforced by the algebraic structure itself.

This feature invites comparison with frameworks of indefinite causal order \cite{Chiribella_2013,Oreshkov_2012}: noncommutativity makes the very notion of a definite measurement ordering problematic, and $E_\circ$ is the self-adjoint expression obtained by symmetrizing the two orderings. The comparison should nevertheless be handled with care: here the symmetrization is not introduced as an operational superposition of causal orders, but as a purely algebraic device to recover an observable expression from noncommuting factors.

Finally, we stress an important limitation for the BMV context. While CHSH-type correlation bounds remain mathematically meaningful in this symmetrized form, the resulting observable $E_\circ$ does not by itself provide an operationally clear entanglement witness in the sense required by BMV: it is not, in general, implementable as a sequence of local measurements nor as a spacelike separated measurement protocol, and we do not assume any measurement prescription for $A\circ B$ beyond formal self-adjointness. Thus, the persistence of Tsirelson's bound does not automatically guarantee the robustness of the specific entanglement witnesses used in BMV proposals once commensurability fails.

This shifts the emphasis of what is experimentally significant. If gravitational dressing obstructs the very definition of standard local correlation observables, then the most direct signature is not an exotic modification of CHSH bounds, but the \emph{violation of commutativity} (and hence of microcausality) itself. In the next section, we therefore turn to the possibility of probing dressing-induced violations of microcausality and to estimating their magnitude in realistic parameter regimes.
\section{Toward experimental bounds on microcausality violation}\label{sec:microcausality_violation}

The previous section showed that relaxing commensurability does not, by itself, lead to stronger-than-quantum correlations: once the CHSH expression is reformulated as an \emph{observable} via the Jordan product, Tsirelson's bound remains $2\sqrt{2}$. The upshot is not a recovery of the standard subsystem picture, but a separation of issues. Correlation \emph{bounds} can be mathematically robust under noncommutativity, whereas the \emph{operational} meaning of entanglement witnesses in BMV-type proposals relies on the availability of commuting subalgebras (and, effectively, a tensor-product subsystem structure) to interpret measured spin correlations as joint local observables.

This motivates a shift in emphasis. If gravitational dressing obstructs microcausality at the algebraic level, then a more direct probe of the underlying structure is the violation of microcausality itself, i.e.\ the appearance of nonvanishing commutators between \emph{gauge-invariant} observables associated with spacelike separated regions. Unlike standard entanglement-witness reasoning, such a diagnostic does not presuppose a tensor factorization or preparation independence.

In what follows, we estimate the size of dressing-induced microcausality violations using parameters relevant for current BMV-type proposals.

\subsection*{Parametric estimate for Coulomb dressing}

We consider the Coulomb dressing introduced in Sec.~\ref{section:Dressed states and Commutativity}. Following \cite{Bose}, we adopt a nonrelativistic regime in which a scalar field $\phi$ can be written as
\begin{equation}
     \phi(x)=\frac{\lambda}{\sqrt{2m}}\,e^{-i m c^2 t/\hbar}\,\psi(x),
\end{equation}
where $\lambda$ is a renormalization constant and $\psi(\mathbf x)\sim (2\pi L^2)^{-3/4}\exp\!\left(-\frac{|\mathbf x|^2}{4L^2}\right)$ is a slowly varying spatial wavefunction of unit norm. In this approximation,
\begin{equation}
    \dot{\phi}(x) = -\frac{i m c^2}{\hbar}\,\phi(x),
    \qquad
    \partial_i \phi(x) \sim \frac{1}{L}\,\phi(x).
\end{equation}

Using the experimental values proposed in \cite{Bose}, namely $L=10^{-6}\,\mathrm{m}$ and $m=10^{-14}\,\mathrm{kg}$, the dressed-field commutators evaluated in Sec.~\ref{section:Dressed states and Commutativity} can be expressed as \emph{relative} estimates by factoring out the local field amplitudes. Using the nonrelativistic scalings $\partial_i\phi \sim L^{-1}\phi$ and, if one keeps the fast rest-energy phase, $\dot\phi \sim (mc^2/\hbar)\phi$, we obtain
\begin{align}
\frac{\big|[\Phi_C(x),\Phi_C(x')]\big|}{|\phi(x)\phi(x')|}
&\sim \frac{G m}{c^2 L}
\;\approx\;10^{-36},
\label{eq:commutator1}
\\[1ex]
\frac{\big|[\dot{\Phi}_C(x),\Phi_C(x')]\big|}{|\phi(x)\phi(x')|}
&\sim \frac{G m}{c^2 L}\,\frac{m c^2}{\hbar}
\;\approx\;
10^{-36}\,\frac{m c^2}{\hbar}
\;\approx\;
10\,\mathrm{s}^{-1},
\label{eq:commutator2}
\end{align}
up to factors of order unity. Here the ratios should be understood as order-of-magnitude estimates for matrix elements in wavepacket states.

The first commutator is therefore extremely suppressed, and moreover can be made to vanish altogether by an appropriate choice of relational observable/dressing \cite{Donnelly:2015hta}, as discussed in Sec.~\ref{section:Dressed states and Commutativity}. At this level one may thus expect the \emph{approximation} of commuting subalgebras---and with it the use of split-type reasoning and product-state idealizations---to remain adequate for many practical purposes, despite being strictly invalid for the fully gauge-invariant dressed observables.

The second commutator is qualitatively different. First, the quantity in Eq.~\eqref{eq:commutator2} is a \emph{rate}: whether it is ``large'' or ``small'' must be assessed relative to a characteristic timescale $\tau$ of the protocol (e.g.\ the interferometric hold time), through the dimensionless combination
\begin{equation}
\epsilon(\tau)\;:=\;\tau\,\frac{\big|[\dot{\Phi}_C(x),\Phi_C(x')]\big|}{|\phi(x)\phi(x')|}\,.
\end{equation}
Second, the order-of-unity-per-second estimate in Eq.~\eqref{eq:commutator2} hinges on keeping the fast rest-energy phase in $\phi$, i.e.\ $\dot\phi\simeq -(imc^2/\hbar)\phi$ with $mc^2/\hbar\simeq 8.5\times 10^{36}\,\mathrm{s}^{-1}$ for $m=10^{-14}\,\mathrm{kg}$. In a strictly nonrelativistic effective description one typically factors out the rest-energy phase and works with slowly varying fields; heuristically, the relevant time dependence is then governed by \emph{slow} frequencies $\omega_{\rm NR}\sim E_{\rm NR}/\hbar$. For instance, a wavepacket of width $L$ has a kinetic scale $E_{\rm kin}\sim \hbar^2/(2mL^2)$ and hence $\omega_{\rm kin}\sim \hbar/(2mL^2)\simeq 5\times 10^{-9}\,\mathrm{s}^{-1}$ for $m=10^{-14}\,\mathrm{kg}$ and $L=10^{-6}\,\mathrm{m}$, which would instead yield
\begin{equation}
\frac{\big|[\dot{\Phi}_C(x),\Phi_C(x')]\big|}{|\phi(x)\phi(x')|}
\sim \frac{Gm}{c^2L}\,\omega_{\rm kin}
\sim 4\times 10^{-44}\,\mathrm{s}^{-1}.
\end{equation}
Even adopting a much larger slow scale $\omega_{\rm NR}\sim 1/\tau$ with $\tau$ of order seconds yields only $\sim 10^{-36}\,\mathrm{s}^{-1}$. In either case, the violation is entirely negligible in magnitude, which is consistent with the excellent validity of linearized quantum gravity in this regime.

Crucially, however, unlike the first commutator, this second commutator cannot in general be made to vanish by a choice of dressing: it reflects the structural nonlocality enforced by the gravitational constraints for gauge-invariant observables. Thus, even if its numerical size is minuscule, it represents a \emph{formal} obstruction to exact microcausality, and hence to the strict microcausal inclusion $\mathcal{A}(\mathcal{O}_A)\subset \mathcal{A}(\mathcal{O}_B)'$ required to even formulate a split inclusion for spacelike separated regions. This leaves open the question of how best to formulate an \emph{effective} split property, and recovering effective notions of statistical independence. 

Finally, we emphasize that microcausality violation in the present sense does not, by itself, necessarily imply superluminal signalling: a weak version of no-signalling can remain satisfied even when microcausality fails. We discuss different notions of no-signalling and their connection to commutativity in App.~\ref{app:nosignalling}.

\section{Discussion}

A central message of this paper is that the standard quantum-information narrative behind BMV-style inference is \emph{not} conceptually automatic once one insists on gauge-invariant observables for gravity. In the bipartite and tripartite idealizations, the two masses are treated as independently addressable subsystems, with well-defined ``local'' operations and (at least approximate) factorization or commensurability. By contrast, in a diffeomorphism-invariant setting, the relevant observables must be gravitationally dressed, and the resulting gauge-invariant subalgebras need not commute even at spacelike separation. This obstructs implementing exact commensurability and blocks the usual route from microcausality to split inclusions, tensor-product factorization, and statistical independence. 

Recent ``quantumness of gravity'' proposals continue to frame the key inference in explicitly or implicitly LOCC terms. In \cite{bose2025spinbasedpathwaytestingquantum}, the authors adopt the familiar entanglement-witness/Bell-style rationale in which the two masses (and their spins) function as independently addressable parties coupled only through a gravitational mediator, while \cite{Lami_2024} goes further by \emph{defining} ``classical gravity'' as dynamics simulable by LOCC and proposing an LOCC inequality test that does not rely on entanglement generation. Meanwhile, \cite{chen2024quantumeffectsgravitynewton} is complementary in spirit: they stress that Newton-potential GIE alone leaves loopholes and propose targeting genuinely field-theoretic quantum signatures tied to the gravitational canonical commutator structure. Nonetheless, in all three cases, the storyline still presupposes an exact subsystem structure. This does not undermine the practical viability of these schemes in the weak-field regime (the induced noncommutativities may be parametrically tiny), but it does qualify the logical force of the step ``observation $\Rightarrow$ gravity is quantum'': the inference is made \emph{within} an effective truncation whose subsystem-independence assumptions are not automatic once diffeomorphism-invariant observables and constraints are treated carefully.

While our main emphasis has been on how gravitational dressing induces violations of microcausality, thereby obstructing exact commensurability and undermining the usual route to the split property and product-state independence, one might still hope that a weaker form of statistical independence remains viable. Raju’s ~\cite{Raju:2021lwh} analysis suggests that this hope might fail. In a diffeomorphism-invariant theory, Gauss-law--type constraints already correlate ``inside'' and ``outside'' data for any bounded region: global charges measured at infinity (notably the ADM charges) are not merely boundary bookkeeping devices but constrain the admissible bulk configurations. Unlike in electromagnetism, where compensating charges can be localized in a collar region to support an effective separation, gravity admits no negative-energy charges that could be used to screen local excitations in this way. As emphasized in~\cite{Raju:2021lwh}, this leads to a sharp obstruction to independent state preparation in full generality: because the gravitational constraints tie bulk data to asymptotic charges, one cannot in general specify fine-grained ``inside'' and ``outside'' data independently in the way the split-property intuition would suggest. If statistical independence breaks down at this level, then the stronger product form of independence invoked in GIE-style inference schemes is \emph{a fortiori} untenable, and any operational use of subsystem independence must be understood as an effective, regime-dependent approximation rather than a fundamental structural feature.

That regime dependence is borne out by our explicit estimates. In the parameter ranges relevant for current BMV-type proposals, the dressing-induced commutators are numerically extremely small, which supports the practical adequacy of modelling the experiment with approximately commuting subalgebras and split-type reasoning at the quantum-information level. At the same time, the field-theoretic analysis shows a qualitative distinction: while some commutator structures can be eliminated by an appropriate choice of dressing, others cannot. Even when their magnitude is minuscule, such nonvanishing commutators therefore represent a formal obstruction to exact microcausality and hence to the strict split-type inclusions that would underwrite product states and subsystem-independence at the fundamental level. This tension motivates treating notions such as ``split property'' and ``independent subsystems'' as \emph{effective} concepts in weak-field gravity, and it leaves open the question of how best to formulate and justify an \emph{effective} split property in general.

In this sense, the present results simultaneously support and qualify the mainstream viewpoint on GIE: they support it by showing that, for realistic parameters, the induced violations of microcausality are tiny enough that the quantum-information modelling can remain an excellent approximation; and they qualify it by stressing that strong ontological conclusions (e.g.\ about ``superposed spacetime'' or fundamental subsystem structure) rely on formal properties---exact factorization, strict commensurability, and robust independence assumptions---that are not preserved once gauge-invariant gravitational observables are taken seriously. More speculatively, the fact that factorization itself is undermined suggests that gravity may affect not only how we \emph{test} quantumness, but also how we should \emph{think} about subsystems and measurement in regimes where diffeomorphism constraints are inescapable~\cite{franzmann2024be}. This strengthens the motivation to complement entanglement-witness logic with more direct, subsystem-agnostic probes of the underlying algebraic structure, such as experimentally bounding dressing-induced violations of microcausality themselves.

Finally, it is worth stressing that the conceptual lessons drawn here are not specific to tabletop gravity-induced entanglement experiments, but extend to other settings where gravity is treated as a perturbative quantized field. In particular, the field-theoretic structures probed in GIE proposals are closely analogous to those underlying the standard treatment of cosmological perturbations during inflation~\cite{Fragkos:2022tbm,franzmann2024be}. In both cases, one works with quantized fluctuations around a classical background and with gauge-invariant, dressed observables, for which strict microcausality and exact Hilbert-space factorization are subtle and, at best, approximate notions; indeed, the same algebraic scrutiny applied here would be expected to yield violations of commutativity that are at least as (if not more) significant in the early Universe regime, given the fully relativistic character and larger amplitudes of the perturbations. This invites a provocative re-reading of the usual ``causal disconnection'' premise behind the horizon problem: if spacelike separated, gauge-invariant observables need not strictly commute, then the sharp subsystem picture in which distinct Hubble patches are taken to be independent is itself an idealization, and part of what inflation is meant to explain---the persistence of correlations across apparently causally disconnected regions---may already be conceptually softened within perturbative quantum gravity. At the same time, and as stressed by our discussion of no-signalling versus commutativity (see App. \ref{app:nosignalling}), nonvanishing spacelike commutators do not by themselves furnish an operational influence channel or an equilibration mechanism, so they might not simply replace inflation \emph{as a mechanism}; rather, they highlight that the locality assumptions underwriting the horizon-problem narrative are effective, regime-dependent approximations whose status should be assessed \textit{within} the same framework used to model cosmological perturbations.

\section*{Acknowledgements}

GF is thankful for financial support from the Olle Engkvist Foundation (no.225-0062) and the Swedish Research Council
(grant number 2022-01893VR). GF is also thankful to all members of the \href{https://emerge-collab.org/}{EmerGe Collaboration}, with whom he learned and discussed many of the formal aspects relevant to this paper. NB is thankful for financial support from the Department of Mathematics and the Council for Doctoral Education (FUR) at Mid Sweden University.

\appendix

\section{Extra theorems and definitions}\label{appendix:extratheorems}

This appendix collects a small number of additional notions used to relate subsystem independence to operational preparation procedures. Throughout, $\alga\subseteq\algb(\hilbert)$ denotes a von Neumann algebra acting on a Hilbert space $\hilbert$, and $\alga'$ its commutant.

\begin{definition}[Factor]\label{def:factors}
A von Neumann algebra $\alga$ is a \emph{factor} iff its center is trivial, i.e.\ $\alga\cap \alga'=\mathbb{C}I$.
\end{definition}

\begin{definition}[Type~I factor]\label{def:typeI}
A factor $\alga$ is \emph{type~I} iff $\alga\simeq \algb(\mathcal K)$ for some Hilbert space $\mathcal K$.
\end{definition}

\begin{definition}[Type~III factor]\label{def:typeIII}
A von Neumann algebra $\mathcal M$ is \emph{type~III} iff it has no nonzero finite projections\footnote{A projection $P\in\mathcal M$ is
\emph{infinite} iff there exists a nonzero subprojection $0\neq Q<P$ such that $Q\sim P$, i.e.\ there is a
partial isometry $V\in\mathcal M$ with $V^*V=P$ and $VV^*=Q$.}.
If moreover $\mathcal M$ is a factor, we call it a \emph{type~III factor}.
Equivalently, for a factor $\mathcal M$ this means that for every nonzero projection $P\in\mathcal M$ there exists
a partial isometry $W\in\mathcal M$ such that $WW^*=P$ and 
$W^*W=\idm$. 
\end{definition}

\noindent
In nonrelativistic quantum mechanics one typically encounters type~I factors, whereas local algebras in
relativistic AQFT are generically type~III \cite{R_dei_2009}.

\medskip

We now recall the operational language used to formulate \emph{local} preparation procedures. An
\emph{operation} on a von Neumann algebra $\alga$ is a normal completely positive (CP) map
$T:\alga\to\alga$ satisfying $T(\idm)\leq \idm $ (subunital). The case $T(\idm)=\idm$ corresponds to a
\emph{nonselective} operation (no post-selection). A state is a positive normalized linear functional
$\omega:\alga\to\mathbb{C}$ with $\omega(\idm)=1$; it is \emph{normal} iff it is ultraweakly continuous (equivalently,
$\omega(A)=\tr(\rho A)$ for some density operator $\rho$ when $\alga=\algb(\hilbert)$).

Operational ``local preparation'' strengthens mere statistical independence: statistical independence only
asks for the \emph{existence} of joint extensions of prescribed marginal states, whereas local preparability
asks that a marginal state can be \emph{implemented} by an operation localized on one subsystem while
leaving the other subsystem unaffected.

\medskip

\begin{definition}[Nonselective operation]\label{def:operation}
A (normal) \emph{nonselective operation} on a vN algebra $\alga$ is a normal completely positive map $T:\alga\to\alga$ satisfying $T(\idm)=\idm$.
\end{definition}

\begin{definition}[Local preparability of a normal state]\label{def:localpreparability}
Let $\alga_1,\alga_2$ be commuting type~III vN factors on a Hilbert space $\hilbert$,
and let $\phi$ be a normal state on $\alga_1$.
We say that $\phi$ is \emph{locally preparable on $\alga_1$ independently of $\alga_2$} if there exists
a normal positive map $T:\mathcal{B}(\hilbert)\to\mathcal{B}(\hilbert)$ such that
\begin{equation}
T(A)=\phi(A)\, T(\idm)\quad\forall\,A\in\alga_1,
\qquad
T(B)= T(\idm)\,B\quad\forall\,B\in\alga_2.
\end{equation}
We say that \emph{some normal state is locally preparable} if there exists at least one normal state
$\phi$ on $\alga_1$ with this property.
\end{definition}

\begin{definition}[Nonselective local preparability of all normal states]\label{def:nonselectivepreparability}
Let $\alga_1,\alga_2$ be commuting type~III vN factors on a Hilbert space $\hilbert$.
We say that \emph{all normal states on $\alga_1$ are nonselectively locally preparable (independently of $\alga_2$)}
if for every normal state $\phi$ on $\alga_1$ there exists a (normal) map
$T:\mathcal{B}(\hilbert)\to\mathcal{B}(\hilbert)$ of the form
\begin{equation}
T(C)=\sum_i W_i^*\, C\, W_i,
\qquad
W_i\in \alga_2',
\end{equation}
such that
\begin{equation}
T(A)=\phi(A)\, T(\idm)\quad\forall\,A\in\alga_1,
\qquad
 T(\idm)=\idm.
\end{equation}
\end{definition}

The following theorem states their equivalence to the split property: 

\begin{theorem}[{\cite[Thm.~5.1]{summers2009}}]\label{th:splitlocalprepnonselprep} Let $\mathcal{A}_1,\mathcal{A}_2$ be commuting type~III vN factors on the Hilbert space $\mathcal{H}$.
Then the following are equivalent:
\begin{enumerate}\itemsep0pt\parskip0pt\parsep0pt\topsep0pt\partopsep0pt 
    \item The pair $(\mathcal{A}_1,\mathcal{A}_2)$ is split;
    \item Local preparability of some normal state;
    \item Nonselective local preparability of all normal states.
\end{enumerate}
\end{theorem}

We include one additional criterion linking uncorrelated states to commutativity, since it cleanly separates ``lack of correlations'' from ``algebraic separability''.

\begin{definition}[Uncorrelated state {\cite[Def.~5.3]{summers2009}}]\label{def:uncorrelated}
A normal state $\phi$ on $\alga_1\vee\alga_2$ is \emph{$\alga_1$--$\alga_2$-uncorrelated} if $\phi(E\wedge F)=\phi(E)\phi(F)$, $\forall$ projections $E\in\alga_1,\ F\in\alga_2.$
\end{definition}

\begin{theorem}[{\cite{Buchholz_2005}}]\label{th:uncorrelated_implies_commute}
Let $\alga_1$ and $\alga_2$ be vN subalgebras of a $W^*$-algebra $\alga$.
If an $\alga_1$--$\alga_2$-uncorrelated state exists on $\alga_1\vee\alga_2$, then $\alga_1$ and $\alga_2$ commute elementwise.
\end{theorem}

\medskip

\tikzset{
  block/.style={
    rectangle, draw, rounded corners=2pt,
    fill=gray!10,
    align=center,
    font=\small,
    text width=10.8em,
    minimum height=3.7em,
    inner sep=3pt,
    drop shadow={opacity=.85, shadow xshift=1.3pt, shadow yshift=-1.3pt}
  },
  rel/.style={-Stealth, thick},
  rel2/.style={Stealth-Stealth, thick},
  lab/.style={
    font=\footnotesize,
    fill=white,
    inner sep=1.2pt,
    fill opacity=0.9,
    text opacity=1
  }
}

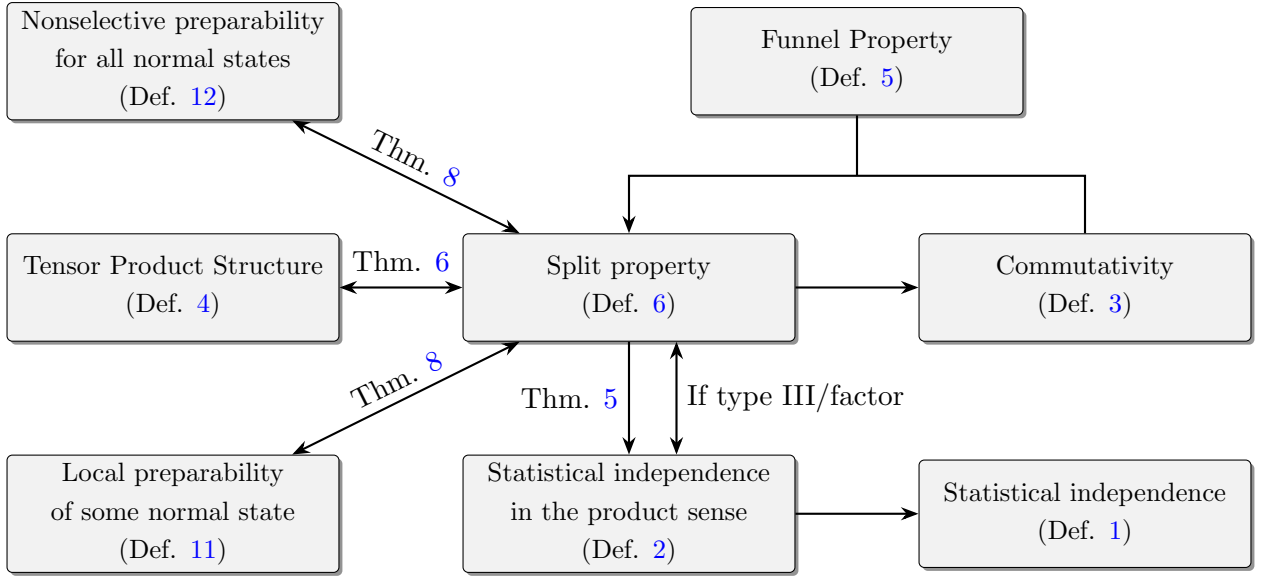
\begin{figure}[H]
  \centering
  \begin{tikzpicture}[node distance = 3cm, auto]
 \node [block] (split) {Split property \\ (Def.  \ref{def:split property})};
     \node [block, left of = split, node distance = 6cm] (structureH) {Tensor Product Structure \\(Def. \ref{def_TPS})};
     \node [block, right of = split, node distance = 6cm] (commensurability) {Commutativity \\ (Def.  \ref{def:Commutativity})};
     \node[block] (funnel) at ($(split)!0.5!(commensurability) + (0,3cm)$)
{Funnel Property \\ (Def. \ref{def:funnel})};
    \node [block, below of = split, node distance = 3cm] (statidpdproduct) {Statistical independence in the product sense \\ (Def. \ref{def:Statistical independence in the product sense})};
   \node [block, right of  = statidpdproduct, node distance = 6cm] (statidpdt) {Statistical independence \\ (Def.  \ref{def:Statistical independence})};

    \node [block, below of = structureH, node distance = 3cm,  text opacity = 1] (locprep) {Local preparability of some normal state \\ (Def.   \ref{def:localpreparability})};
    
    \node [block, above of = structureH, node distance = 3cm,  text opacity = 1] (nonsel) {Nonselective preparability for all normal states \\ (Def.   \ref{def:nonselectivepreparability})};
   
 \draw[rel2] (split) -- node[midway, above=2pt, sloped] {Thm. \ref{th:equivalentpairsplitspatialisomorphism}}(structureH);
    \draw [rel] (split) --(commensurability);
    \draw [rel] (split) --node[midway, left] {Thm. \ref{th:ifsplitthenopidpdt}} (statidpdproduct);
    \draw [rel2]([xshift=0.625cm]split.south) --node[midway, right] {If type III/factor} ([xshift=0.625cm]statidpdproduct.north);
      \draw [rel] (statidpdproduct) -- (statidpdt);
\coordinate (fjoin) at ($(funnel.south)+(0,-0.8cm)$);
\draw[thick] (funnel.south) -- (fjoin);
\draw[thick] (fjoin) -| (commensurability.north);
\draw[rel] (fjoin) -| (split.north);
 \draw [rel2] (split) --node[midway, above, sloped] {Thm. \ref{th:splitlocalprepnonselprep}}(locprep);
    
    \draw [rel2] (split) --node[midway, above, sloped] {Thm. \ref{th:splitlocalprepnonselprep}}(nonsel);
  \end{tikzpicture}
  \caption{Summary of all the properties, and their respective relations, introduced in the text and appendix. }
  \label{graph:big}
\end{figure}

\section{No-signalling versus commutativity of subalgebras}
\label{app:nosignalling}

In this appendix, we clarify the logical relation between no-signalling and commutativity of subsystem algebras in quantum theory. The essential point is that \emph{no-signalling} admits inequivalent operational formulations of different strength, depending on which local interventions are taken to be physically admissible. While elementwise commutativity of subalgebras is sufficient for all reasonable notions of no-signalling, it is necessary only for the strongest (``maximal'') ones. This distinction is crucial in constrained gauge theories, and in particular for gravity, where the set of physically implementable operations may be strictly smaller than the set of algebraically definable operations.

\vspace{0.5cm}

Let $\mathcal H$ be a Hilbert space and let $\mathcal A_1,\mathcal A_2\subset\mathcal B(\mathcal H)$ be von Neumann subalgebras, representing the observables accessible to two parties \emph{Alice} and \emph{Bob}. We identify normal states with density operators $\rho\in\mathcal T_1(\mathcal H)$ (positive trace-class operators with $\tr\rho=1$), so that expectation values are $\langle X\rangle_\rho=\tr(\rho X)$ for $X\in\mathcal B(\mathcal H)$.

\subsection*{Operational no-signalling and restricted local interventions}
We denote by $\mathrm{Ops}_\mathrm{A}(\mathcal A_1)$ a chosen set of \emph{physically admissible} local operations for Alice. Concretely, we take elements of $\mathrm{Ops}_\mathrm{A}(\mathcal A_1)$ to be normal completely positive, trace-preserving (CPTP) maps
$\mathcal E_\mathrm{A}:\mathcal T_1(\mathcal H)\to\mathcal T_1(\mathcal H)$.
A standard \emph{sufficient} algebraic criterion for ``localization on $\mathcal A_1$'' is that $\mathcal E_\mathrm{A}$ admits a Kraus decomposition with Kraus operators in $\mathcal A_1$:
\begin{equation}
\label{eq:krausA}
\mathcal E_\mathrm{A}(\rho)=\sum_i K_i\rho K_i^\dagger,
\qquad
K_i\in\mathcal A_1,
\qquad
\sum_i K_i^\dagger K_i=\idm.
\end{equation}
In unconstrained quantum mechanics, one often idealizes $\mathrm{Ops}_\mathrm{A}(\mathcal A_1)$ as \emph{all} such Kraus maps. In constrained systems (e.g.\ gauge theories), this identification is generally too strong: not every algebraically definable Kraus map corresponds to a physically admissible intervention, so $\mathrm{Ops}_\mathrm{A}(\mathcal A_1)$ should be regarded as a (possibly strict) subset\footnote{Equivalently, in the Heisenberg picture one may characterize ``local'' operations by requiring that the dual map acts trivially on the commutant, $\mathcal E_\mathrm{A}^\ast(X)=X$ for all $X\in\mathcal A_1'$. Kraus locality \eqref{eq:krausA} implies this, but the physically admissible set may still be smaller.}.

We denote by $\mathcal E_\mathrm{A}^\ast:\mathcal B(\mathcal H)\to\mathcal B(\mathcal H)$ the normal unital dual map defined by
$\tr(\mathcal E_\mathrm{A}(\rho)X)=\tr(\rho\,\mathcal E_\mathrm{A}^\ast(X))$ for all $\rho\in\mathcal T_1(\mathcal H)$ and $X\in\mathcal B(\mathcal H)$.

\begin{definition}[Operational no-signalling]
We say that \emph{operational no-signalling from Alice to Bob} holds (relative to $\mathrm{Ops}_\mathrm{A}(\mathcal A_1)$) if, for all normal states $\rho\in\mathcal T_1(\mathcal H)$, all admissible CPTP maps $\mathcal E_\mathrm{A}\in \mathrm{Ops}_\mathrm{A}(\mathcal A_1)$, and all $X\in \mathcal A_2$,
\begin{equation}
\label{eq:nosignalling}
\tr(\rho\,X)=\tr\!\big(\mathcal E_\mathrm{A}(\rho)\,X\big).
\end{equation}
\end{definition}

This notion of no-signalling is explicitly relative to the allowed class of local operations. In particular, if $\mathrm{Ops}_\mathrm{A}(\mathcal A_1)$ is severely restricted, then \eqref{eq:nosignalling} may hold even when the algebras fail to commute.

\paragraph{Commutativity implies operational no-signalling, but the converse fails.}
If $[\mathcal A_1,\mathcal A_2]=0$, then Eq.~\eqref{eq:nosignalling} holds for all states and for all admissible local operations that admit a Kraus decomposition with $K_i\in\mathcal A_1$. Indeed, if $\mathcal E_\mathrm{A}$ admits Kraus operators $\{K_i\}\subseteq\mathcal A_1$ as in \eqref{eq:krausA}, then $[K_i,X]=0$ for all $X\in \mathcal A_2$, and hence
\[
\tr\!\big(\mathcal E_\mathrm{A}(\rho)\,X\big)
=\sum_i \tr(\rho\,K_i^\dagger X K_i)
=\sum_i \tr(\rho\,X K_i^\dagger K_i)
=\tr(\rho\,X).
\]
Commutativity is therefore a sufficient condition for operational no-signalling.

The converse implication fails in general: operational no-signalling \eqref{eq:nosignalling} does not, by itself, imply commutativity. The reason is that \eqref{eq:nosignalling} constrains only those disturbances that can be effected by \emph{admissible} operations in $\mathrm{Ops}_\mathrm{A}(\mathcal A_1)$. Even if $[\mathcal A_1,\mathcal A_2]\neq 0$, signalling may be impossible because the physically admissible operations attributed to Alice are too restricted to change expectation values of Bob-observables in $\mathcal A_2$.

\begin{ex}
Let $\mathcal H=\mathcal H_{\mathrm{A},L}\otimes\mathcal H_{\mathrm{A},R}\otimes\mathcal H_{\mathrm{B}}$ and define
\[
\mathcal A_1=\mathcal B(\mathcal H_{\mathrm{A},L}\otimes\mathcal H_{\mathrm{A},R})\otimes\idm_{\mathrm{B}},
\qquad
\mathcal A_2=\idm_{\mathrm{A},L}\otimes U_{\mathrm{A},R\,\mathrm{B}}\big(\idm_{\mathrm{A},R}\otimes\mathcal B(\mathcal H_{\mathrm{B}})\big)U_{\mathrm{A},R\,\mathrm{B}}^\dagger,
\]
where $U_{\mathrm{A},R\,\mathrm{B}}$ is an entangling unitary on $\mathcal H_{\mathrm{A},R}\otimes\mathcal H_{\mathrm{B}}$. In general,
$[\mathcal A_1,\mathcal A_2] \neq 0$ because operators acting nontrivially on $\mathcal H_{\mathrm{A},R}$ fail to commute with $\mathcal A_2$.

Now suppose that Alice's \emph{admissible control} is restricted to the tensor factor $\mathcal H_{\mathrm{A},L}$, i.e.\ take $\mathrm{Ops}_\mathrm{A}(\mathcal A_1)$ to consist of CPTP maps of the form
$\mathcal E_\mathrm{A}=\Phi_{\mathrm{A},L}\otimes\idm_{\mathrm{A},R}\otimes\idm_{\mathrm{B}}$.
Then $\mathcal E_\mathrm{A}^\ast(X)=X$ for all $X\in \mathcal A_2$, and hence
\[
\tr(\rho\,X)=\tr\!\big(\mathcal E_\mathrm{A}(\rho)\,X\big)
\qquad \forall\,\rho,\ \forall\,X\in\mathcal A_2.
\]
Thus, noncommuting subalgebras can coexist with operational no-signalling once the admissible class of ``local'' operations is restricted, for instance by having the \emph{algebra of accessible observables} and the \emph{algebra of operationally controllable interventions} not coincide.
\end{ex}

Elementwise commutativity expresses \emph{commensurability}, whereas operational no-signalling expresses the absence of causal influence \emph{relative to the physically admissible} class of local interventions. The latter need not imply the former in general, unless one assumes a sufficiently rich (in a sense, \emph{maximal}) class of local operations. This motivates considering a strictly stronger notion of no-signalling, often called \emph{macrocausality}, defined by invariance under all non-selective sharp measurements.

\subsection*{Macrocausality and maximal no-signalling}

Given a (finite or countable) projection-valued measure (PVM) $\{P_i\}_i\subseteq\mathcal A_1$, the associated non-selective L\"uders update of a normal state $\rho$ is
\begin{equation}
\label{eq:luders}
\rho\longmapsto\rho':=\sum_i P_i\rho P_i,
\end{equation}
with the sum understood in trace norm (equivalently ultraweakly on expectation values).
We say that \emph{Alice cannot signal to Bob by non-selective measurements} if for every normal state $\rho$, every PVM $\{P_i\}\subseteq\mathcal A_1$, and every projection $Q\in\mathcal A_2$,
\begin{equation}
\label{eq:nosig_proj}
\tr(\rho\,Q)=\tr(\rho'\,Q).
\end{equation}

\begin{theorem}[\cite{Malament1996-MALIDO,EmerGe_proj_2}]
\label{thm:nosig_comm}
Condition~\eqref{eq:nosig_proj} holds for all normal states, all PVMs in $\mathcal A_1$, and all projections in $\mathcal A_2$ if and only if
\[
[\mathcal A_1,\mathcal A_2]=0.
\]
\end{theorem}

The equivalence in the theorem relies on a \emph{maximal} operational assumption: invariance under all non-selective L\"uders measurements associated with arbitrary PVMs (equivalently, under a very rich family of local interventions). While mathematically natural, this assumption is too strong for constrained gauge theories and gravity, where not every algebraic measurement corresponds to a physically admissible local intervention. With respect to the weaker operational notion of no-signalling in Eq.~\eqref{eq:nosignalling} (relative to the chosen set $\mathrm{Ops}_\mathrm{A}(\mathcal A_1)$), violations of commutativity induced by gravitational dressing do not \emph{necessarily} imply superluminal signalling, but rather the breakdown of algebraic separability between subsystems.

\phantomsection
\addcontentsline{toc}{section}{References}

\let\oldbibliography\thebibliography
\renewcommand{\thebibliography}[1]{
  \oldbibliography{#1}
  \setlength{\parskip}{0pt}
  \setlength{\itemsep}{0pt} 
  \footnotesize 
}
\bibliographystyle{JHEP2}
\bibliography{bib}

@article{Marletto,
  title = {Gravitationally Induced Entanglement between Two Massive Particles is Sufficient Evidence of Quantum Effects in Gravity},
  author = {Marletto, C. and Vedral, V.},
  journal = {Phys. Rev. Lett.},
  volume = {119},
  issue = {24},
  pages = {240402},
  numpages = {5},
  year = {2017},
  month = {Dec},
  publisher = {American Physical Society},
  doi = {10.1103/PhysRevLett.119.240402},
}

@article{DiBiagio:2025twt,
    author = "Di Biagio, Andrea",
    title = "{The simple reason why classical gravity can entangle}",
    eprint = "2511.02683",
    archivePrefix = "arXiv",
    primaryClass = "gr-qc",
    month = "11",
    year = "2025"
}

@article{Raju:2021lwh,
    author = "Raju, Suvrat",
    title = "{Failure of the split property in gravity and the information paradox}",
    eprint = "2110.05470",
    archivePrefix = "arXiv",
    primaryClass = "hep-th",
    doi = "10.1088/1361-6382/ac482b",
    journal = "Class. Quant. Grav.",
    volume = "39",
    number = "6",
    pages = "064002",
    year = "2022"
}

@misc{Bose_Supplementary_2017,
  title        = {Supplementary Material for {A Spin Entanglement Witness for Quantum Gravity}},
  author       = {Bose, Sougato and Mazumdar, Anupam and Morley, Gavin W. and Ulbricht, Hendrik and Toroš, Marko and Paternostro, Mauro and Geraci, Andrew and Barker, Peter and Kim, M. S. and Milburn, Gerard},
  howpublished = {Supplementary Material to Phys. Rev. Lett. 119, 240401 (2017)},
  year         = {2017},
  note         = {\url{https://journals.aps.org/prl/supplemental/10.1103/PhysRevLett.119.240401}},
}

@article{Zanardi:2004zz,
    author = "Zanardi, Paolo and Lidar, Daniel A. and Lloyd, Seth",
    title = "{Quantum tensor product structures are observable induced}",
    eprint = "quant-ph/0308043",
    archivePrefix = "arXiv",
    doi = "10.1103/PhysRevLett.92.060402",
    journal = "Phys. Rev. Lett.",
    volume = "92",
    pages = "060402",
    year = "2004"
}

@article{Fragkos:2022tbm,
    author = "Fragkos, Vasileios and Kopp, Michael and Pikovski, Igor",
    title = "{On inference of quantization from gravitationally induced entanglement}",
    eprint = "2206.00558",
    archivePrefix = "arXiv",
    primaryClass = "quant-ph",
    doi = "10.1116/5.0101334",
    journal = "AVS Quantum Sci.",
    volume = "4",
    pages = "045601",
    year = "2022"
}

@article{Huggett:2022uui,
    author = "Huggett, Nick and Linnemann, Niels and Schneider, Mike",
    title = "{Quantum Gravity in a Laboratory?}",
    eprint = "2205.09013",
    archivePrefix = "arXiv",
    primaryClass = "quant-ph",
    month = "5",
    year = "2022"
}

@article{Galley:2020qsf,
    author = "Galley, Thomas D. and Giacomini, Flaminia and Selby, John H.",
    title = "{A no-go theorem on the nature of the gravitational field beyond quantum theory}",
    eprint = "2012.01441",
    archivePrefix = "arXiv",
    primaryClass = "quant-ph",
    doi = "10.22331/q-2022-08-17-779",
    journal = "Quantum",
    volume = "6",
    pages = "779",
    year = "2022"
}

@article{Lami_2024, title={Testing the Quantumness of Gravity without Entanglement}, volume={14}, ISSN={2160-3308}, DOI={10.1103/physrevx.14.021022}, number={2}, journal={Physical Review X}, publisher={American Physical Society (APS)}, author={Lami, Ludovico and Pedernales, Julen S. and Plenio, Martin B.}, year={2024}, month=may }

@misc{ludescher2025gravitymediatedentanglementinfinitedimensionalsystems,
      title={Gravity-mediated entanglement via infinite-dimensional systems}, 
      author={Stefan L. Ludescher and Leon D. Loveridge and Thomas D. Galley and Markus P. Müller},
      year={2025},
      eprint={2507.13201},
      archivePrefix={arXiv},
      primaryClass={quant-ph},
}

@article{Christodoulou_2019,
   title={On the possibility of laboratory evidence for quantum superposition of geometries},
   volume={792},
   ISSN={0370-2693},
   DOI={10.1016/j.physletb.2019.03.015},
   journal={Physics Letters B},
   publisher={Elsevier BV},
   author={Christodoulou, Marios and Rovelli, Carlo},
   year={2019},
   month=may, pages={64–68} }

@article{Cotler:2017abq,
    author = "Cotler, Jordan S. and Penington, Geoffrey R. and Ranard, Daniel H.",
    title = "{Locality from the Spectrum}",
    eprint = "1702.06142",
    archivePrefix = "arXiv",
    primaryClass = "quant-ph",
    reportNumber = "SU-ITP-17-01",
    doi = "10.1007/s00220-019-03376-w",
    journal = "Commun. Math. Phys.",
    volume = "368",
    number = "3",
    pages = "1267--1296",
    year = "2019"
}

@article{Christodoulou:2022mkf,
    author = "Christodoulou, Marios and Di Biagio, Andrea and Aspelmeyer, Markus and Brukner, {\v{C}}aslav and Rovelli, Carlo and Howl, Richard",
    title = "{Locally Mediated Entanglement in Linearized Quantum Gravity}",
    eprint = "2202.03368",
    archivePrefix = "arXiv",
    primaryClass = "quant-ph",
    doi = "10.1103/PhysRevLett.130.100202",
    journal = "Phys. Rev. Lett.",
    volume = "130",
    number = "10",
    pages = "100202",
    year = "2023"
}

@article{Wallace:2021qyh,
    author = "Wallace, David",
    title = "{Quantum Gravity at Low Energies}",
    eprint = "2112.12235",
    archivePrefix = "arXiv",
    primaryClass = "gr-qc",
    month = "12",
    year = "2021"
}

@article{Oppenheim:2018igd,
    author = "Oppenheim, Jonathan",
    title = "{A Postquantum Theory of Classical Gravity?}",
    eprint = "1811.03116",
    archivePrefix = "arXiv",
    primaryClass = "hep-th",
    doi = "10.1103/PhysRevX.13.041040",
    journal = "Phys. Rev. X",
    volume = "13",
    number = "4",
    pages = "041040",
    year = "2023"
}

@incollection{Malament1996-MALIDO,
	author = {David Malament},
	booktitle = {Perspectives on Quantum Reality},
	editor = {R. Clifton},
	pages = {35--136},
	publisher = {Kluwer Academic Publishers},
	title = {In Defense of Dogma: Why There Cannot Be a Relativistic Quantum Mechanical Theory of (Localizable) Particles},
	year = {1996}
}

@article{R_dei_2009,
   title={When Are Quantum Systems Operationally Independent?},
   volume={49},
   ISSN={1572-9575},
   DOI={10.1007/s10773-009-0010-5},
   number={12},
   journal={International Journal of Theoretical Physics},
   publisher={Springer Science and Business Media LLC},
   author={Rédei, Miklós and Summers, Stephen J.},
   year={2009},
   month=may, pages={3250–3261} }

@article{Donnelly:2015hta,
    author = "Donnelly, William and Giddings, Steven B.",
    title = "{Diffeomorphism-invariant observables and their nonlocal algebra}",
    eprint = "1507.07921",
    archivePrefix = "arXiv",
    primaryClass = "hep-th",
    reportNumber = "NSF-KITP-15-133",
    doi = "10.1103/PhysRevD.93.024030",
    journal = "Phys. Rev. D",
    volume = "93",
    number = "2",
    pages = "024030",
    year = "2016",
    note = "[Erratum: Phys.Rev.D 94, 029903 (2016)]"
}

@article{Bose,
  title = {Spin Entanglement Witness for Quantum Gravity},
  author = {Bose, Sougato and Mazumdar, Anupam and Morley, Gavin W. and Ulbricht, Hendrik and Toro\ifmmode \check{s}\else \v{s}\fi{}, Marko and Paternostro, Mauro and Geraci, Andrew A. and Barker, Peter F. and Kim, M. S. and Milburn, Gerard},
  journal = {Phys. Rev. Lett.},
  volume = {119},
  issue = {24},
  pages = {240401},
  numpages = {6},
  year = {2017},
  month = {Dec},
  publisher = {American Physical Society},
  doi = {10.1103/PhysRevLett.119.240401},
}

@article{Chitambar:2014svb,
    author = "Chitambar, Eric and Leung, Debbie and Man\v{c}inska, Laura and Ozols, Maris and Winter, Andreas",
    title = "{Everything You Always Wanted to Know About LOCC (But Were Afraid to Ask)}",
    doi = "10.1007/s00220-014-1953-9",
    journal = "Commun. Math. Phys.",
    volume = "328",
    epint = "quant-ph/1210.4583",
    archivePrefix = "arXiv",
    number = "1",
    pages = "303--326",
    year = "2014"
}

@article{Horodecki:2009zz,
    author = "Horodecki, Ryszard and Horodecki, Pawel and Horodecki, Michal and Horodecki, Karol",
    title = "{Quantum entanglement}",
    eprint = "quant-ph/0702225",
    archivePrefix = "arXiv",
    doi = "10.1103/RevModPhys.81.865",
    journal = "Rev. Mod. Phys.",
    volume = "81",
    pages = "865--942",
    year = "2009"
}

@misc{franzmann2024be,
      title={To be or not to be, but where?}, 
      author={Guilherme Franzmann},
      year={2024},
      eprint={2405.21031},
      archivePrefix={arXiv},
      primaryClass={quant-ph}
}

@misc{summers2009,
      title={Subsystems and Independence in Relativistic Microscopic Physics}, 
      author={Stephen J. Summers},
      year={2009},
      eprint={0812.1517},
      archivePrefix={arXiv},
      primaryClass={quant-ph},
}

@article{summers1990,
author = {Summers, Stephen J.},
title = {On the independence of local algebras in quantum field theory},
journal = {Reviews in Mathematical Physics},
volume = {02},
number = {02},
pages = {201-247},
year = {1990},
doi = {10.1142/S0129055X90000090} }

@article{cirelson_1980,
  title={Quantum Generalizations of Bell’s Inequality.},
  author={Cirel’son, B. S.},
  journal={Letters in Mathematical Physics},
  volume={4},
  number={2},
  pages={93-100},
  year={1980},
  doi={https://doi.org/10.1007/BF00417500}
}

@misc{redei2006quantumprobabilitytheory,
      title={Quantum Probability Theory}, 
      author={Miklos Redei and Stephen J. Summers},
      year={2006},
      eprint={quant-ph/0601158},
      archivePrefix={arXiv},
      primaryClass={quant-ph},
}

@article{EmerGe_proj_2,
  author       = {Colafranceschi, Eugenia and Di Biagio, Andrea and Flinckman, Joakim and Franzmann, Guilherme and Glowacki, Jan and Linnemann, Niels and Niedermann, Florian},
  collaboration = {EmerGe Collaboration},
  title        = {Subsystems Independence: from classical mechanics to quantum field theory and beyond},
  year         = {2026},
  eprint       = {in preparation},
  note         = {\url{https://emerge-collab.org}}
}

@book{KadisonRingrose1986,
  title     = {Fundamentals of the Theory of Operator Algebras, Volume II: Advanced Theory},
  author    = {Richard V. Kadison and John R. Ringrose},
  publisher = {Academic Press},
  address   = {Orlando, FL},
  year      = {1986},
  series    = {Pure and Applied Mathematics, Vol. 100\,II},
  isbn      = {0-12-393302-7},
}

@article{Terhal_2000,
   title={Bell inequalities and the separability criterion},
   volume={271},
   ISSN={0375-9601},
   DOI={10.1016/s0375-9601(00)00401-1},
   number={5–6},
   journal={Physics Letters A},
   publisher={Elsevier BV},
   author={Terhal, Barbara M.},
   year={2000},
   month=jul, pages={319–326} }

@article{doplicher_standard_1984,
	title = {Standard and split inclusions of von {Neumann} algebras},
	volume = {75},
	issn = {1432-1297},
	doi = {10.1007/BF01388641},
	number = {3},
	journal = {Inventiones mathematicae},
	author = {Doplicher, S. and Longo, R.},
	month = oct,
	year = {1984},
	pages = {493--536},
}

@article{Buchholz_2005,
   title={Quantum statistics and locality},
   volume={337},
   ISSN={0375-9601},
   DOI={10.1016/j.physleta.2005.01.055},
   number={1–2},
   journal={Physics Letters A},
   publisher={Elsevier BV},
   author={Buchholz, Detlev and Summers, Stephen J.},
   year={2005},
   month=mar, pages={17–21} }

@article{Bose_2022,
   title={Mechanism for the quantum natured gravitons to entangle masses},
   volume={105},
   ISSN={2470-0029},
   DOI={10.1103/physrevd.105.106028},
   number={10},
   journal={Physical Review D},
   publisher={American Physical Society (APS)},
   author={Bose, Sougato and Mazumdar, Anupam and Schut, Martine and Toroš, Marko},
   year={2022},
   month=may }

@article{Donnelly:2016rvo,
    author = "Donnelly, William and Giddings, Steven B.",
    title = "{Observables, gravitational dressing, and obstructions to locality and subsystems}",
    eprint = "1607.01025",
    archivePrefix = "arXiv",
    primaryClass = "hep-th",
    doi = "10.1103/PhysRevD.94.104038",
    journal = "Phys. Rev. D",
    volume = "94",
    number = "10",
    pages = "104038",
    year = "2016"
}

@article{PhysicsPhysiqueFizika.1.195,
  title = {On the Einstein Podolsky Rosen paradox},
  author = {Bell, J. S.},
  journal = {Physics Physique Fizika},
  volume = {1},
  issue = {3},
  pages = {195--200},
  numpages = {6},
  year = {1964},
  month = {Nov},
  publisher = {American Physical Society},
  doi = {10.1103/PhysicsPhysiqueFizika.1.195},
}

@article{Aziz:2025ypo,
    author = "Aziz, Joseph and Howl, Richard",
    title = "{Classical theories of gravity produce entanglement}",
    eprint = "2510.19714",
    archivePrefix = "arXiv",
    primaryClass = "quant-ph",
    doi = "10.1038/s41586-025-09595-7",
    month = "10",
    year = "2025"
}

@ARTICLE{1969PhRvL..23..880C,
       author = {{Clauser}, John F. and {Horne}, Michael A. and {Shimony}, Abner and {Holt}, Richard A.},
        title = "{Proposed Experiment to Test Local Hidden-Variable Theories}",
      journal = {Physical Review Letters},
         year = 1969,
        month = oct,
       volume = {23},
       number = {15},
        pages = {880-884},
          doi = {10.1103/PhysRevLett.23.880},
       adsurl = {https://ui.adsabs.harvard.edu/abs/1969PhRvL..23..880C}
}

@article{PhysRevD.14.1944,
  title = {Proposed experiment to test the nonseparability of quantum mechanics},
  author = {Aspect, Alain},
  journal = {Phys. Rev. D},
  volume = {14},
  issue = {8},
  pages = {1944--1951},
  numpages = {0},
  year = {1976},
  month = {Oct},
  publisher = {American Physical Society},
  doi = {10.1103/PhysRevD.14.1944},
}

@article{PhysRevLett.49.1804,
  title = {Experimental Test of Bell's Inequalities Using Time-Varying Analyzers},
  author = {Aspect, Alain and Dalibard, Jean and Roger, G\'erard},
  journal = {Phys. Rev. Lett.},
  volume = {49},
  issue = {25},
  pages = {1804--1807},
  numpages = {0},
  year = {1982},
  month = {Dec},
  publisher = {American Physical Society},
  doi = {10.1103/PhysRevLett.49.1804},
}

@article{Oreshkov_2012,
   title={Quantum correlations with no causal order},
   volume={3},
   ISSN={2041-1723},
   DOI={10.1038/ncomms2076},
   number={1},
   journal={Nature Communications},
   publisher={Springer Science and Business Media LLC},
   author={Oreshkov, Ognyan and Costa, Fabio and Brukner, Caslav},
   year={2012},
   month=oct }

@article{Chiribella_2013,
   title={Quantum computations without definite causal structure},
   volume={88},
   ISSN={1094-1622},
   DOI={10.1103/physreva.88.022318},
   number={2},
   journal={Physical Review A},
   publisher={American Physical Society (APS)},
   author={Chiribella, Giulio and D’Ariano, Giacomo Mauro and Perinotti, Paolo and Valiron, Benoit},
   year={2013},
   month=aug }

@misc{bose2025spinbasedpathwaytestingquantum,
      title={A Spin-Based Pathway to Testing the Quantum Nature of Gravity}, 
      author={Sougato Bose and Anupam Mazumdar and Roger Penrose and Ivette Fuentes and Marko Toroš and Ron Folman and Gerard J. Milburn and Myungshik Kim and Adrian Kent and A. T. M. Anishur Rahman and Cyril Laplane and Aaron Markowitz and Debarshi Das and Ethan Campos-Méndez and Eva Kilian and David Groswasser and Menachem Givon and Or Dobkowski and Peter Skakunenko and Maria Muretova and Yonathan Japha and Naor Levi and Omer Feldman and Damián Pitalúa-García and Jonathan M. H. Gosling and Ka-Di Zhu and Marco Genovese and Kia Romero-Hojjati and Ryan J. Marshman and Markus Rademacher and Martine Schut and Melanie Bautista-Cruz and Qian Xiang and Stuart M. Graham and James E. March and William J. Fairbairn and Karishma S. Gokani and Joseph Aziz and Richard Howl and Run Zhou and Ryan Rizaldy and Thiago Guerreiro and Tian Zhou and Jason Twamley and Chiara Marletto and Vlatko Vedral and Jonathan Oppenheim and Mauro Paternostro and Hendrik Ulbricht and Peter F. Barker and Thomas P. Purdy and M. V. Gurudev Dutt and Andrew A. Geraci and David C. Moore and Gavin W. Morley},
      year={2025},
      eprint={2509.01586},
      archivePrefix={arXiv},
      primaryClass={quant-ph},
}

@misc{chen2024quantumeffectsgravitynewton,
      title={Quantum effects in gravity beyond the Newton potential from a delocalised quantum source}, 
      author={Lin-Qing Chen and Flaminia Giacomini},
      year={2024},
      eprint={2402.10288},
      archivePrefix={arXiv},
      primaryClass={quant-ph},
}

@article{DelicEtAl2020LevitatedCooling,
  author  = {Deli{\'c}, U. and Reisenbauer, M. and Dare, K. and Grass, D. and Vuleti{\'c}, V. and Kiesel, N. and Aspelmeyer, M.},
  title   = {Cooling of a levitated nanoparticle to the motional quantum ground state},
  journal = {Science},
  volume  = {367},
  pages   = {892},
  year    = {2020}
}

@article{WestphalEtAl2021MillimeterGravity,
  author        = {Westphal, T. and Hepach, H. and Pfaff, J. and Aspelmeyer, M.},
  title         = {Measurement of Gravitational Coupling between Millimeter-Sized Masses},
  journal       = {Nature},
  volume        = {591},
  pages         = {225},
  year          = {2021},
  eprint        = {2009.09546},
  archivePrefix = {arXiv}
}

@article{MargalitEtAl2021SternGerlachInterferometer,
  author        = {Margalit, Y. and Dobkowski, O. and Zhou, Z. and Amit, O. and Japha, Y. and Moukouri, S. and Rohrlich, D. and Mazumdar, A. and Bose, S. and Henkel, C. and Folman, R.},
  title         = {Realization of a complete Stern-Gerlach interferometer: Towards a test of quantum gravity},
  journal       = {Science Advances},
  volume        = {7},
  pages         = {eabg2879},
  year          = {2021},
  eprint        = {2011.10928},
  archivePrefix = {arXiv}
}

@misc{Aspelmeyer2022AvoidClassicalWorld,
  author        = {Aspelmeyer, M.},
  title         = {How to avoid the appearance of a classical world in gravity experiments},
  year          = {2022},
  eprint        = {2203.05587},
  archivePrefix = {arXiv}
}

@article{PandaEtAl2024LatticeAtomInterferometerGravity,
  author  = {Panda, C. D. and Tao, M. J. and Ceja, J. and Khoury, J. and Tino, G. M. and M{\"u}ller, H.},
  title   = {Measuring gravitational attraction with a lattice atom interferometer},
  journal = {Nature},
  volume  = {631},
  pages   = {515},
  year    = {2024}
}

@article{DANTONI1983361,
title = {Interpolation by type I factors and the flip automorphism},
journal = {Journal of Functional Analysis},
volume = {51},
number = {3},
pages = {361-371},
year = {1983},
issn = {0022-1236},
doi = {https://doi.org/10.1016/0022-1236(83)90018-6},
author = {Claudio D'Antoni and Roberto Longo}
}

@article{Kibble:1968sfb,
    author = "Kibble, T. W. B.",
    title = "{Coherent Soft-Photon States and Infrared Divergences. I. Classical Currents}",
    doi = "10.1063/1.1664582",
    journal = "J. Math. Phys.",
    volume = "9",
    number = "2",
    pages = "315--324",
    year = "1968"
}

@article{Kulish:1970ut,
    author = "Kulish, P. P. and Faddeev, L. D.",
    title = "{Asymptotic conditions and infrared divergences in quantum electrodynamics}",
    reportNumber = "D70-07927",
    doi = "10.1007/BF01066485",
    journal = "Theor. Math. Phys.",
    volume = "4",
    pages = "745",
    year = "1970"
}

@article{DonnellyFreidel2016LocalSubsystems,
  author        = {Donnelly, William and Freidel, Laurent},
  title         = {Local subsystems in gauge theory and gravity},
  journal       = {JHEP},
  volume        = {09},
  pages         = {102},
  year          = {2016},
  doi           = {10.1007/JHEP09(2016)102},
  eprint        = {1601.04744},
  archivePrefix = {arXiv},
  primaryClass  = {hep-th}
}

@article{Giddings2019SoftChargesSplitting,
  author        = {Giddings, Steven B.},
  title         = {Gravitational dressing, soft charges, and perturbative gravitational splitting},
  journal       = {Phys. Rev. D},
  volume        = {100},
  pages         = {126001},
  year          = {2019},
  doi           = {10.1103/PhysRevD.100.126001},
  eprint        = {1903.06160},
  archivePrefix = {arXiv},
  primaryClass  = {hep-th}
}

@article{Giddings_2006,
   title={Observables in effective gravity},
   volume={74},
   ISSN={1550-2368},
   DOI={10.1103/physrevd.74.064018},
   number={6},
   journal={Physical Review D},
   publisher={American Physical Society (APS)},
   author={Giddings, Steven B. and Marolf, Donald and Hartle, James B.},
   year={2006},
   month=sep }

@article{Bagan_2000,
   title={Charges from Dressed Matter: Construction},
   volume={282},
   ISSN={0003-4916},
   DOI={10.1006/aphy.2000.6048},
   number={2},
   journal={Annals of Physics},
   publisher={Elsevier BV},
   author={Bagan, Emili and Lavelle, Martin and McMullan, David},
   year={2000},
   month=jun, pages={471–502} }

@article{Torre:1993fq,
    author = "Torre, C. G.",
    title = "{Gravitational observables and local symmetries}",
    eprint = "gr-qc/9306030",
    archivePrefix = "arXiv",
    reportNumber = "FTG-116-USU",
    doi = "10.1103/PhysRevD.48.R2373",
    journal = "Phys. Rev. D",
    volume = "48",
    pages = "R2373--R2376",
    year = "1993"
}

@article{Marolf:2015jha,
    author = "Marolf, Donald",
    title = "{Comments on Microcausality, Chaos, and Gravitational Observables}",
    eprint = "1508.00939",
    archivePrefix = "arXiv",
    primaryClass = "gr-qc",
    doi = "10.1088/0264-9381/32/24/245003",
    journal = "Class. Quant. Grav.",
    volume = "32",
    number = "24",
    pages = "245003",
    year = "2015"
}

@article{Giddings:2018cjc,
    author = "Giddings, Steven B.",
    title = "{Quantum gravity: a quantum-first approach}",
    eprint = "1805.06900",
    archivePrefix = "arXiv",
    primaryClass = "hep-th",
    doi = "10.31526/LHEP.3.2018.01",
    journal = "LHEP",
    volume = "1",
    number = "3",
    pages = "1--3",
    year = "2018"
}

@article{deBoer:2022zka,
    author = "de Boer, Jan and others",
    title = "{Frontiers of Quantum Gravity: shared challenges, converging directions}",
    eprint = "2207.10618",
    archivePrefix = "arXiv",
    primaryClass = "hep-th",
    month = "7",
    year = "2022"
}

@article{EmerGe_proj_1,
  author       = {Colafranceschi, Eugenia and Di Biagio, Andrea and Flinckman, Joakim and Franzmann, Guilherme and Glowacki, Jan and Linnemann, Niels and Niedermann, Florian},
  collaboration = {EmerGe Collaboration},
  title        = {Microcausality and low-energy quantum gravity},
  year         = {2026},
  eprint       = {in preparation},
  note         = {\url{https://emerge-collab.org}}
}

\end{document}